\documentclass[11pt]{article}

\RequirePackage{amsthm,amsmath,amsfonts,amssymb}
\RequirePackage[numbers,sort&compress]{natbib}
\RequirePackage[colorlinks,citecolor=blue,urlcolor=blue]{hyperref}

\usepackage[utf8]{inputenc}
\usepackage{siunitx}
\usepackage{verbatim, float, dsfont, graphicx,caption, subcaption, tabularx, booktabs, mathrsfs, url, multicol,multirow,xcolor}
\usepackage{algorithmic}
\usepackage{algorithm}
\usepackage{enumerate}
\usepackage{yk}

\graphicspath{{Fig/}}

\newcommand{\indep}{\perp \!\!\! \perp}

\newcommand{\VC}{\mathrm{VC}}
\newcommand{\DNN}{\mathrm{DNN}}
\newcommand{\trunc}{\mathrm{trunc}}
\newcommand{\pre}{\mathrm{pre}}

\def\r#1{\textcolor{red}{\bf #1}}
\def\b#1{\textcolor{blue}{\bf #1}}

\definecolor{NavyBlue}{rgb}{0.4, 0.6, 0.8}

\begin{document}


\title{\bf UTOPIA: Universally Trainable Optimal Prediction Intervals Aggregation}
\author{Jianqing Fan$^\dagger$\thanks{
    This work is partially supported by ONR grant N00014-19-1-2120 and the NSF grants DMS-2052926, DMS-2053832, and DMS-2210833}, Jiawei Ge$^\dagger$ and Debarghya Mukherjee$^\ddagger$\hspace{.2cm}\\\\
    $^\dagger$Department of Operations Research and Financial Engineering, \\ Princeton University\\
    $^\ddagger$Department of Mathematics and Statistics, \\ 
    Boston University}
\date{}

\maketitle

\begin{abstract}
Uncertainty quantification in prediction presents a compelling challenge with vast applications across various domains, including biomedical science, economics, and weather forecasting. 
There exists a wide array of methods for constructing prediction intervals, such as quantile regression and conformal prediction. 
However, practitioners often face the challenge of selecting the most suitable method for a specific real-world data problem.
In response to this dilemma, we introduce a novel and universally applicable strategy called Universally Trainable Optimal Predictive Intervals Aggregation (UTOPIA). 
This technique excels in \emph{efficiently aggregating} multiple prediction intervals while maintaining a small average width of the prediction band and ensuring coverage. 
UTOPIA is grounded in linear or convex programming, making it straightforward to train and implement.  In the specific case where the prediction methods are elementary basis functions, as in kernel and spline bases, our method becomes the construction of a prediction band.
Our proposed methodologies are supported by theoretical guarantees on the coverage probability and the average width of the aggregated prediction interval, which are detailed in this paper. 
The practicality and effectiveness of UTOPIA are further validated through its application to synthetic data and two real-world datasets in finance and macroeconomics.
\end{abstract}

\noindent%

\section{Introduction}
Uncertainty quantification plays a crucial role in a wide range of scenarios where decision-makers require accurate point prediction and the uncertainty quantification associated with it. 
For instance, in finance, a stock price prediction interval can help investors make informed decisions and manage market risks, while in health care, a disease progression prediction interval is important for making treatment decisions and improving patient outcomes. 
Furthermore,  reliable prediction intervals with provable coverage guarantees that quantify uncertainty in machine learning-based predictive models are now in high demand. 
In addition to providing a high-level coverage guarantee, a prediction interval with a small width is also important because it conveys more precise information about the uncertainties in prediction, enabling better decision-making.

Several attempts have been made to construct prediction intervals owing to the importance of uncertainty quantification, e.g., quantile regression \citep{koenker1978regression, koenker2001quantile, belloni2019conditional}, conformal prediction \cite{lei2014distribution, barber2022conformal}, universal prediction band \citep{liang2021universal}, predictive inference \citep{duchi2024predictive}, prediction interval for regression and classification \cite{chen2021learning} etc., each with its own advantages. 
It is often challenging for a practitioner to select the most appropriate prediction interval from the various available options.
This motivates the following question: 
\begin{center}
\textit{Can we efficiently aggregate various methods to create a prediction interval that ensures both coverage and a small width?}    
\end{center}
To address this, we propose a method, UTOPIA, for aggregating prediction intervals via a simple linear or convex programming problem that directly aims to minimize the expected width of the aggregated prediction interval, while maintaining the coverage guarantee. 
More specifically, given $K$ prediction intervals $\{(\ell_i(X), u_i(X))\}_{1 \le i \le K}$, our goal is to find weights $\{\alpha_i,\beta_i\}^K_{i=1}$ to form an effective prediction interval $(\sum_i \alpha_i \ell_i(X), \sum_i \beta_i u_i(X))$ with small expected width and sufficient coverage.
These individual intervals can be derived from any method, such as linear or non-linear quantile regression, conformal methods, or others. 
Empirical experiments show that our aggregated approach often outperforms individual methods.

Our approach draws inspiration from the extensive literature on model aggregation, whose theoretical foundations were established in seminal works such as \cite{nemirovski2000topics, juditsky2000functional}, and have been further developed in subsequent studies, including \cite{yang2003regression, yang2004aggregating, tsybakov2004optimal,rigollet2007linear}.
The core principle of model aggregation lies in identifying optimal weights so that a weighted combination of multiple models performs on par with, or better than, the best individual model in the collection.  
Aggregation of various point estimates has gained significant attraction in the regression problem, as evidenced in studies like \cite{rigollet2012kullback,dalalyan2012sparse,dalalyan2012sharp,lecue2014optimal}, among others.
Similarly, the aggregation of distributional estimates has also captured the attention of the forecasting community, as seen in works such as \cite{raftery2005using,kapetanios2015generalised,rasp2018neural}.
Despite these advancements, the domain of prediction interval aggregation remains relatively unexplored and not as well understood.
There has been a recent surge in research focused on aggregating quantile estimates. \cite{nowotarski2015computing} introduced a method known as Quantile Regression Averaging (QRA), which employs quantile linear regression on the outputs from individual base models.
This concept has been further expanded upon in several ways.
For instance, \cite{maciejowska2016probabilistic} applied PCA to reduce the dimensionality of base models before implementing the QRA aggregator; \cite{fakoor2021flexible} explored the idea of using weights that vary not only across individual models but also across different quantile levels and feature values.
Another line of research (e.g., Cross-Conformal Predictors and Bootstrap Conformal Predictors by \cite{vovk2015cross}, and Aggregated Conformal Predictors by \cite{carlsson2014aggregated}) focuses on aggregating conformal predictors.  These methods aim to maximize data utilization through data splitting or bagging techniques, though they may not always produce the smallest prediction region.
Motivated by conformal prediction, \cite{yang2021finite} put forward two selection
algorithms that aim to obtain the smallest conformal
prediction region given a family of point prediction machine learning algorithms. 

In comparison to the previous methods, our aggregation method stands out for its flexibility, as it doesn't rely solely on one type of estimate. We can incorporate a variety of estimators, including those based on quantile regression, conformal methods, or RKHS (as proposed in \cite{liang2021universal}).
Our technique is universally applicable to any set of prediction intervals, regardless of their specific structure. 
We furthermore demonstrate that our method can also be used to construct a prediction band from a general function class $\cF$ beyond model aggregation. 
One distinguished feature of the obtained prediction interval is that it has coverage guarantees even in cases of model misspecification in $\cF$. 
Additionally, if any of the available prediction intervals to aggregate are noisy and produce unusually large intervals, UTOPIA typically assigns them smaller weights. 
In other words, UTOPIA is robust against incorporating such noisy prediction intervals.
To summarize, in this work, we propose an approach that: 
\begin{enumerate}
    \item \ Efficiently aggregates various prediction intervals to achieve a small average width while ensuring a provable coverage guarantee.
    \item \ Is computationally efficient, involving only solving a linear or quadratic optimization problem. 
    \item \ Can be employed to construct a prediction interval from a general function class $\cF$ that possesses appealing theoretical properties.
    \item \ Allow in-depth statistical analysis of its coverage and optimality.
    \item \ Is demonstrated through various applications both in synthetic and real data. 
\end{enumerate}
As illustrated in our simulation studies, the aggregation allows ensemble learning from multiple methods and distributes computation into each individual method.  
UTOPIA emphasizes universality and trainability due to its properties 1 and 2, though it is also a testable method with the ability to construct prediction intervals from an elementary basis.  

The rest of the paper is organized as follows: 
Section \ref{sec:methods} outlines our general methodology for constructing prediction intervals. 
Section \ref{sec:theory} offers theoretical guarantees for our approach. 
Section \ref{sec:applications} discusses the theoretical aspects of several important special cases within our general method. 
Section \ref{simulation} includes extensive simulations to validate our theoretical results and to provide comparisons with existing methods. 
Additionally, Section \ref{sec:real_data_analysis} showcases the practical utility of our method through applications to two real datasets, the Fama-French data, and the FRED-MD data. For brevity, detailed proofs are included in the Appendix.

\section{UTOPIA Method}
\label{sec:methods}
We assume that we have a pair of random variables $(X, Y) \sim P$, where $X\in\cX$ and $Y\in\cY$. 
In this paper we typically assume $\cX = \reals^p$ and $\cY = \reals$. 
Our goal is to construct a prediction interval of the form $[l(X), u(X)]$ of $Y$ given $X$ based on $n$ i.i.d random observations $(X_1, Y_1), \dots, (X_n, Y_n)$ from $P$. 
Here we aim for \emph{expected guarantee} instead of \emph{conditional guarantee}, i.e. $\bbP(\ell(X) \leq Y \leq u(X)) \ge 1 -\alpha$ where the probability is taken over the joint distribution of $(X, Y)$. 
Constructing any such prediction interval is easy; given any positive function $f$, choose $c \equiv c(f)$ such that: 
$$
\bbP(-cf(X) \le Y \le cf(X)) \ge 1 - \alpha \,.
$$
Such a $c$ always exists as soon as the distribution of $(X, Y)$ is continuous since the probability on the left-hand side goes to $0$ as $c$ goes to $0$, goes to $1$ as $c$ goes to $\infty$. 
In other words, any positive function can be converted to a prediction interval with $(1-\alpha)$ coverage guarantee just by a simple re-scaling. 
However, this procedure is not efficient as some choices of $f$ may lead to unnecessarily large prediction intervals. 
Therefore, one should ideally construct a prediction interval with a small width and adequate coverage, which leads to the following optimization problem: 
\begin{equation}\label{alpha_quantile}
    \begin{aligned}
    \min_{l, u} \ & \bbE[u(X) - l(X)] \\
    \st \ & \bbP(Y \in [l(X), u(X)]) \ge 1 - \alpha \,.
    \end{aligned}
\end{equation}
The interval is constructed around $Y$ without any specific reference point, making the problem potentially complex due to the variability in the distribution of $Y$.
In the context of many machine learning models, particularly regression models, the primary objective is to minimize the mean squared error, which positions the conditional mean $m_0(x) = \bbE[Y \mid X = x]$ as a pivotal element in prediction.
To simplify the model and make it more tractable, we assume the distribution of $Y$ is symmetric around $m_0(X)$ almost surely on $X$.
This symmetry introduces a specific reference point, namely $m_0(X)$, around which $Y$ is distributed, and the variance of $Y$ around $m_0(X)$ is effectively captured by the squared deviation $(Y-m_0(X))^2$.
The task then is to find a function $f$ that bounds this deviation with a certain probability $1-\alpha$ while being minimal in expected value. 
Mathematically speaking, this can be expressed as follows: For a given function class $\cF$ and the centering function $m_0$:
\begin{equation}\label{population_opt:m_known}
    \begin{aligned}
        \min_{f \in \cF} & \ \ \ \ \bbE[f(X)] \\
        \st & \  \ \ \  \bbP\left(f(X)\geq(Y - m_0(X))^2\right)\geq 1-\alpha.
    \end{aligned}
\end{equation}
The key advantage of \eqref{population_opt:m_known} over \eqref{alpha_quantile} is that we now need to optimize one function $f$ instead of $(\ell, u)$ if $m_0$ is known. 
When $m_0$ is unknown we can estimate it using any standard technique (parametric or non-parametric) of conditional mean estimation. 
However, as will be elaborated later, our techniques can be easily generalized to the situation when the distribution of $Y$ is not symmetric around $m_0(X)$, i.e. when we optimize over both $(\ell, u)$. 

The solution of \eqref{population_opt:m_known} is the \emph{optimal} prediction interval in the sense that it has the smallest average width among all the prediction intervals with $(1 - \alpha)$ average guarantee. However, the above optimization problem is hard to solve directly; for example when we have $n$ samples from the distribution of $(X, Y)$, then it is tempting to replace the objective $\bbE[\cdot]$ by sample average over $X$ and the constraint $\bbP(\cdot)$ by empirical proportion, i.e. solve the following optimization problem: for given function classes $\cF$ for the shape of width and $\cM$ for the centering functions, 
\begin{equation}
\label{eq:sample_opt_indicator}
    \begin{aligned}
        \min_{f \in \cF, m \in \cM} & \ \ \ \ \frac1n \sum_i f(x_i) \\
        \st & \  \ \ \ \frac1n \sum_i \mathds{1}_{\left\{f(x_i) \ge (y_i - m(x_i))^2\right\}} \geq 1-\alpha.
    \end{aligned}
\end{equation}
While this approach may have theoretical appeal, its computational complexity arises due to the non-convex nature of the constraint. More specifically the constraint \newline $\bbP\left(f(X)\geq(Y - m_0(X))^2\right)$ is a non-convex functional on $\cF$, yielding a computationally intractable problem. To alleviate the issue, we propose the following two-step approach, inspired by \cite{liang2021universal}: 
\\\\
{\bf Step 1 [Initial Estimation]: }Given a collection of observation $\{(x_i, y_i)\}_{1 \le i \le n}$, obtain an initial estimate $\hat f_{\pre}$ and $\hat m$ by solving the following: 
\begin{equation}
\label{sample_opt:general_form}
    \begin{aligned}
        \min_{ m\in \cM, f \in \cF} & \ \ \ \ \frac1n \sum^n_{i=1}f(x_i) \\
        \st & \  \ \ \  f(x_i)\geq(y_i - m(x_i))^2 \ \ \ \forall \ 1 \le i \le n \,.
    \end{aligned}
\end{equation}

\noindent{\bf Step 2 [Re-scaling]: } 
Upon obtaining $\hat f_\pre$, re-scale it appropriately so that the resulting function has desired coverage: given any $\alpha > 0$, find $\hat \lambda(\alpha)$ such that: 
$$
\hat \lambda(\alpha) = \inf_{\lambda}\left\{\lambda: \frac{1}{n} \sum_{i=1}^n\mathbf{1}_{\{(y_i-\hat m(x_i))^2>\lambda \hat f_\pre(x_i)\}}\leq \alpha\right\} \,.
$$
and set the prediction interval as: 
$$
\widehat{\mathrm{PI}}_{(1-\alpha)}(x) = \left[\hat m(x)-\sqrt{\hat \lambda(\alpha)\hat f_\pre(x)}, \hat m(x)-\sqrt{\hat \lambda(\alpha)\hat f_\pre(x)}\right] \,.
$$
Algorithm~\ref{algo:main_algo} furnishes additional details when the function class is a non-VC class (to be defined later).

\noindent
\begin{remark}
   In Step 1, we require $f(x)$ to be greater than $(y - m(x))^2$ for all samples, but this can be affected by outliers when $Y - m_0(X)$ has heavy tails. To address this, we suggest truncating (Winsorizing) observations by setting $\tilde{y} = \text{median}(-\tau, y, \tau)$ and solving the optimization problem using $\tilde{y_i}$ instead of $y_i$. This approach might result in a slightly wider interval, but as our analysis will show, it will not affect the coverage. The minimal width guarantee can be maintained by selecting a truncation level $\tau_n$ that increases slowly with $n$, based on the tail behavior of the conditional distribution.
\end{remark}

In the first step, we consider a computationally tractable sample-based optimization problem, where the constraints ensure that the learned function $\hat f_{\pre}$ and $\hat m$ is a feasible solution of the intractable problem \eqref{eq:sample_opt_indicator}. 
Furthermore, Step 1 is computationally efficient because when $\cF$ is a convex collection of functions (e.g., a set of all functions that are $k$ times differentiable), \eqref{sample_opt:general_form} is a convex optimization problem and can be solved using any standard solver.  
For an example, consider the prediction interval aggregation problem where $\cF = \{\sum_{j = 1}^K \beta_j f_j\}$ for some given $\{f_1, \dots, f_K\}$ and also assume $m_0$ is known. Then \eqref{sample_opt:general_form} simplifies to: 
\begin{equation}
\label{eq:aggregation}
    \begin{aligned}
        \min_{\beta_1, \dots, \beta_K} & \ \ \ \ \frac1n \sum^n_{i=1}\sum_{j = 1}^K\beta_j f_j(x_i) \\
        \st & \  \ \ \  \sum_{j = 1}^K \beta_j f_j(x_i)\geq(y_i - m_0(x_i))^2 \ \ \ \forall \ 1 \le i \le n \,.
    \end{aligned}
\end{equation}
It is immediate that the above problem is a linear optimization problem over the unknown parameters $(\beta_1, \dots, \beta_K)$ and easy to implement. We elaborate on this in Section \ref{sec:aggregation}. 
Although we are minimizing the bandwidth explicitly in \eqref{sample_opt:general_form}, the resulting prediction band will typically be conservative (because it covers all the observed samples).
Therefore, in the second step, we shrink $\hat f_\pre$ with an additional factor $\hat \lambda(\alpha)$ while maintaining the desired empirical coverage of $1-\alpha$. The second step is also very easy to implement as $\hat \lambda(\alpha)$ is nothing but $(1-\alpha)$ quantile of $(y_i - \hat m(x_i))/\hat f_\pre(x_i)$. Our methodology is cogently presented in Algorithm \ref{algo:main_algo}. To summarize 
\begin{algorithm}[!t]
   \caption{UTOPIA: Universally Trainable Optimal  Prediction Intervals \\ (Aggregation with coverage probability $1-\alpha$)}
\label{algo:main_algo}
\begin{algorithmic}[1]
    \REQUIRE{$\alpha$, sample $\cD = \{(x_i, y_i)\}_{i = 1}^{n}$, function classes $\cM, \cF$, and $\delta$, where $\delta = 0$ when $(\cF, \cM)$ are VC class of functions and some small number for more complicated function classes.}
    \STATE Use the data to solve \eqref{sample_opt:general_form} and obtain $\hat m, \hat f_{\pre}$. 
    \STATE Set $\hat \lambda({\alpha})$ as the $(1-\alpha)^{th}$ empirical quantile of $(y_i -\hat m(x_i))^2/(\hat f_{\pre}(x_i)+\delta) $, i.e. 
    \begin{align*}
    &\hat \lambda({\alpha}):=\inf_{\lambda}\left\{\lambda: \frac{1}{n} \sum_{i=1}^n\mathbf{1}_{\{(y_i-\hat m(x_i))^2>\lambda(\hat f_{\pre}(x_i)+\delta)\}}\leq \alpha\right\} \,.
    \end{align*} 
    \STATE Return $(\hat \lambda({\alpha}), \hat m, \hat f_{\pre})$ with the estimated prediction interval: 
    $$
    \widehat{\mathrm{PI}}_{(1-\alpha)}(x,\delta) = \left[\hat m(x)-\sqrt{\hat \lambda({\alpha})(\hat f_{\pre}(x)+\delta)}, \ \hat m(x)+\sqrt{\hat \lambda({\alpha}) (\hat f_{\pre}(x)+\delta)}\right] \,.
    $$
\end{algorithmic}
\end{algorithm}
our method has three significant advantages applicable to many real-life scenarios: 
\begin{enumerate}
    \item \textit{Computational efficiency}: Step 1 is a convex optimization problem: it simplifies to linear or quadratic programming (when mean function $m_0$ is also aggregated by linear combinations) problem for prediction interval aggregation, and Step 2 merely requires finding order statistics of $n$ observations. Therefore, our method is computationally efficient. This is why we emphasize the trainability in UTOPIA, though our method is testable in theory and applications. 
    
    \item \textit{Coverage guarantee}: We have theoretical guarantees for the coverage of our estimated prediction interval (See Section \ref{sec:coverage_guarantee}), \textit{regardless of the correct or incorrect specifications of the model}.
    \item \textit{Average width guarantee}: Last but not least, our method directly minimizes the average width (see \eqref{sample_opt:general_form}) and consequently yields a prediction interval with adequate coverage and small average width as evident from our theories and real/synthetic experiments. 
\end{enumerate}

In the Step 1 of our approach, we estimate $(f, m)$ simultaneously. However, in many statistical problems, either we have some available estimator for the conditional mean function or estimating it is significantly easier (say when the conditional mean function is linear). In that case, we can further simplify our approach by the following two-step procedure: 
i) First estimate $m_0$ from $\cM$ via parametric or non-parametric regression, and then 
    ii) plugin $\hat m$ in \eqref{sample_opt:general_form} to obtain $\hat f_\pre$. We summarize this procedure in the Algorithm \ref{algo:main_algo_twostep}. This method has a computational advantage as it splits optimization with two unknown functions into two separate functions learning, and often, estimation of computation mean function is easy, e.g., using kernel-based methods, spline-based approaches, or deep neural networks. 

\begin{algorithm}[!t]
   \caption{Two-step UTOPIA}
    \label{algo:main_algo_twostep}
\begin{algorithmic}[1]
    \REQUIRE{$\alpha$, sample $\cD = \{(x_i, y_i)\}_{i = 1}^{n}$, function classes $\cM, \cF$, and $\delta$, where $\delta = 0$ when $(\cF, \cM)$ are VC class of functions and some small number for more complicated function classes.}
    \STATE Use the data to first estimate $m_0$ via non-parametric regression, i.e. 
    $$
    \hat m = \argmin_{m \in \cM} \left[\sum_{i \in \cD} (y_i - m(x_i))^2 + \rho_n(m) \right] \,,
    $$
    where $\rho_n(\cdot)$ is a penalty function. 
    \STATE Use the data to solve \eqref{sample_opt:general_form} with $\hat m$ fixed and obtain $\hat f_{\pre}$. 
    \STATE Compute the $(1-\alpha)^{th}$ empirical quantile of $(y_i -\hat m(x_i))^2/(\hat f_{\pre}(x_i)+\delta) $:
    \begin{align*}
&\hat \lambda({\alpha}):=\inf_{\lambda}\left\{\lambda: \frac{1}{n} \sum_{i=1}^n\mathbf{1}_{\{(y_i-\hat m(x_i))^2>\lambda(\hat f_{\pre}(x_i)+\delta)\}}\leq \alpha\right\} \, .
    \end{align*}
    \STATE \textbf{Return} $(\hat \lambda({\alpha}), \hat m, \hat f_{\pre})$ with the estimated prediction interval: 
    $$
    \widehat{\mathrm{PI}}_{(1-\alpha)}(x,\delta) = \left[\hat m(x)-\sqrt{\hat \lambda({\alpha}) (\hat f_{\pre}(x)+\delta)}, \ \hat m(x)+\sqrt{\hat \lambda({\alpha})(\hat f_{\pre}(x)+\delta)}\right] \,.
    $$ 
\end{algorithmic}
\end{algorithm}

\begin{remark}[Extension to non-symmetric distribution]
\label{rem:asymmetry}
    Although we have assumed the distribution of $Y \mid X$ to be symmetric around its mean $m_0$, our method can also be extended to asymmetric distribution. In that case, if we are only interested in estimating a prediction interval of $Y$ given $X$, not its mean, then we can modify \eqref{sample_opt:general_form} as follows: 
    \begin{equation}
        \begin{aligned}
            \min_{f_1 \in \cF_1, f_2 \in \cF_2} & \ \ \frac1n \sum_i(f_2(x_i) - f_1(x_i)) \\
            \st & \ \  f_1(x_i) \le y_i \le f_2(x_i), \ \ \forall \ 1 \le i \le n \,.
        \end{aligned}
    \end{equation}
Our theories can be extended easily to incorporate this situation. 
\end{remark}

\begin{remark}[Extend the prediction interval]
If the underlying function class is complex (e.g. non-VC class), then often sample level guarantee does not immediately translate to the population level guarantee. In these cases, we recommend slightly extending the prediction interval
for some small $\delta > 0$ (as shown in Algorithm \ref{algo:main_algo} and Algorithm \ref{algo:main_algo_twostep}) that depends on $n$ and typically decreases to $0$ as $n \to \infty$. For more details, see Remark \ref{rmk:cov}. 
\end{remark}

\vspace{-.2in}
\subsection{Exmaple 1: Aggregation of prediction bands}
\label{sec:aggregation}
In this example, we elaborate on the prediction band aggregation problem as demonstrated in \eqref{eq:aggregation}. 
One may use several different techniques for constructing prediction intervals, i.e. various parametric or non-parametric methods (e.g., quantile regression, conformal prediction). 
Additionally, modern machine learning techniques like neural networks can be employed to create prediction intervals by utilizing their ability to capture complex patterns in data. 
Suppose we have $\cM_0 = \{m_1, \dots, m_L\}$ and $\cF_0 = \{f_1, \dots, f_K\}$, i.e. we have some candidate estimates for estimating the mean function and the prediction band respectively. In this subsection, we aim to construct a suitable linear combination of the elements of $\cM$ and $\cF$ so that the resulting prediction interval has adequate coverage and small average width. 
Then \eqref{sample_opt:general_form} reduces to the following quadratic optimization problem: 
\begin{equation}\label{sample_opt:aggregation}
    \begin{aligned}
        \min_{
        \substack{\alpha_1, \dots, \alpha_K \\ \beta_1, \dots, \beta_L}
        } & \ \ \ \ \frac1n \sum^n_{i=1}  \Bigl \{\sum_{j=1}^K\alpha_jf_j(x_i) \Bigr \} \\
        \st & \  \ \ \ \sum_{j=1}^K\alpha_j f_j(x_i) \geq \Bigl (y_i - \Bigl \{\sum_{j=1}^L \beta_j m_j(x_i) \Bigr \} \Bigr )^2 \ \ \ \forall \ 1 \le i \le n \,, \\
         & \ \ \ \ \alpha_j \ge 0, \ \ \ \forall \ 1 \le j \le K \,.
    \end{aligned}
\end{equation}
The first curly bracket indicates the ``squared width'' at point $X_i$ and the optimization object is the average ``squared width''. The second curly bracket adjusts the center of the predictive intervals to make the width as small as possible. 
In matrix notation, if we define $\bM_n \in \reals^{n \times L}$ and $\bF_n \in \reals^{n \times K}$ with $\bM_{n, ij} = m_j(x_i)$ and $\bF_{n, ij} = f_j(x_i)$, Then we can rewrite \eqref{sample_opt:aggregation} as
\begin{equation*}
    \begin{aligned}
        \min_{\balpha \in \reals^K_+, \bbeta \in \reals^L} & \ \ \ \mathbf{1}^{\top} \bF_n \balpha \\
        \st & \  \ \ \  e_i^\top \left(\bY - \bM_n \bbeta\right)\left(\bY - \bM_n \bbeta\right)^\top e_i - e_i^\top \bF_n \balpha  \le 0\ \ \ \forall \ 1 \le i \le n \,.
    \end{aligned}
\end{equation*}
As the constraint function is quadratic in $\bbeta$ and linear in $\balpha$, it is jointly convex in $(\balpha, \bbeta)$ and consequently can be solved efficiently via standard convex optimization techniques. Let $\widehat{\balpha}$ and $\widehat{\bbeta}$ be the solution to the above optimization problem. Then our initial estimator $\hat f_\pre = \sum_j \hat \alpha_j f_j$ and $\hat m = \sum_j \hat \beta_j m_j$. If $\hat \lambda(\alpha)$ is obtained via Step 2, i.e. $(1-\alpha)$ quantile of $(y_i - \hat m(x_i))^2/\hat f_\pre(x_i)$, then our final prediction interval is:
\begin{align}
\label{PI_aggregation}
&\widehat{\mathrm{PI}}_{(1 - \alpha)}(x) \notag\\
& = \left[\sum_{j = 1}^L \hat \beta_j m_j(x) - \sqrt{\hat \lambda(\alpha)\sum_{j =1}^K \hat \alpha_j f_j(x)},\ \ \sum_{j = 1}^L \hat \beta_j m_j(x) + \sqrt{\hat \lambda(\alpha)\sum_{j = 1}^K \hat \alpha_j f_j(x)}\right] \,.
\end{align}
The theoretical guarantees (both in terms of the average width and the coverage guarantee) of the aggregated prediction interval elaborated above are provided in Section \ref{sec:model_aggregation}.

\subsection{Example 2: Prediction bands via Kernel Basis}
\label{sec:rkhs}
In this subsection, we estimate the mean function and the prediction band via kernel basis assuming they lie on a reproducing kernel Hilbert space. 
For simplicity, we assume $m_0$ is known and aim to approximate optimal $f$ (that yields minimum width with adequate coverage) by some non-negative function over an RKHS $\cH$ with kernel $K$ and feature map $\Phi$, where typically we take $\Phi(x) = K(x, \cdot)$. As $f$ is non-negative, it is tempted to approximate $f$ by $\cF = \{\sum_j \alpha_j K(x_j, \cdot): \alpha_j \ge 0\}$. However, as elaborated in \cite{marteau2020non}, they may not be universal approximators in the sense that they may not approximate all non-negative functions and the authors suggested using the quadratic form $\{\langle \Phi(x), \cA(\Phi)(x)\rangle_\cH\}$ where $\cA$ is bounded linear Hermitian operator from $\cH \mapsto \cH$ and $\langle \cdot, \cdot \rangle_\cH$ is the inner product associate with the Hilbert space $\cH$. 
\cite{liang2021universal} proposed the following method to estimate the prediction band: 
\begin{equation}
\begin{aligned}
    \label{eq:liang}
    \min_{\cA: \cH\rightarrow\cH} \ \ & \|\cA\|_{\star} \\
    \st \ \ & \langle \Phi(x_i), \cA(\Phi)(x_i)\rangle_\cH \ge (y_i - m_0(x_i))^2 \\
    &\cA \succeq 0\,,
\end{aligned}
\end{equation}
where $\| \cdot \|_{\star}$ denotes the nuclear norm of an operator. The aim of this objective is to minimize the rank of $\cA$. However as our objective is to minimize the expected empirical width, we start with a simpler function class $\cF = \{f(x) = \langle \Phi(x), \cA(\Phi)(x)\rangle_\cH, \|\cA\|_{\star} \le r, \cA \succeq 0 \}$ for some given rank $r$ and consequently \eqref{sample_opt:general_form} boils down to the following: 
\begin{equation}
\begin{aligned}
    \label{eq:liang_modified}
    \min_{\cA: \cH\rightarrow\cH} \ \ &  \frac1n \sum_i \langle \Phi(x_i), \cA(\Phi)(x_i)\rangle_\cH  \\
    \st \ \ & \langle \Phi(x_i), \cA(\Phi)(x_i)\rangle_\cH \ge (y_i - m_0(x_i))^2  \\
    & \|\cA\|_\star \le r , \cA \succeq 0\,.
\end{aligned}
\end{equation}

By representer theorem in \cite{marteau2020non}, the solution to the above infinite-dimensional problem can be written as $\cA = \sum^n_{i,j=1}B_{ij}\Phi(x_i)\otimes\Phi(x_j)$ for some positive semi-definite matrix $B$, which implies an equivalent finite-dimensional problem:
\begin{equation}
\begin{aligned}
    \label{eq:finite_SDP}
    \min_{B\in\bbS^{n\times n}} \ \ & \tr(KBK) \\
    \st \ \ & \left \langle K_i, BK_i\right \rangle  \ge (y_i - m_0(x_i))^2 \ \forall \ 1 \le i \le n \\
    & \tr(KB) \le r, \ \ B \succeq 0 \,.
\end{aligned}
\end{equation}
Here $K\in\bbS^{n\times n}$ is the corresponding kernel matrix with $K_{ij}=K(x_i,x_j)$ and $K_i\in\bbR^n$ represents the $i$-th column of the matrix $K$. Notice that \eqref{eq:finite_SDP} is semi-definite programming as we have the cone constraint ($B \succeq 0$) and both the objectives and constraints are linear in $B$. Therefore, it can be easily implemented via a standard convex programming package (e.g. cvx in MATLAB). Let $\hat B$ be the solution to the above optimization problem. The our initial estimator $\hat f_\pre(x) = \langle K_x, \hat B K_x \rangle$ where $K_x\in\bbR^n$ with $K_{x,i}=K(x_i,x)$. Next, we obtain $\hat \lambda (\alpha)$ as the $(1-\alpha)^{th}$ quantile of $(y_i - m_0(x_i))^2/\hat f_\pre(x_i)$ and set our prediction interval as: 
\begin{align*}
\widehat{\mathrm{PI}}_{(1 - \alpha)}(x) & = \left[m_0(x) - \sqrt{\hat \lambda (\alpha) \langle K_x, \hat BK_x \rangle}, \ \ m_0(x) + \sqrt{\hat \lambda (\alpha) \langle K_x, \hat BK_x\rangle}\right] \,,
\end{align*}
If we do not know $m_0$, we can use standard RKHS regression techniques to obtain an estimate $\hat m$ and replace $m_0$ by $\hat m$ in the above prediction interval. 
For more details and the theoretical guarantees on the average width and coverage of this prediction interval see Section \ref{RKHS}.

\section{Statistical Properties}
\label{sec:theory}


In this section, we delve into the theoretical underpinnings of UTOPIA, which hinge on the complexity of the underlying function classes. 
To set the stage for our in-depth analysis, we first introduce the foundational concepts related to the complexity of the function classes that will be referenced throughout this discussion.

\subsection{Preliminaries}
We use $\cD$ to denote $n$ observed samples. We assume there exists $\cB_{\cM}>0$ and $\cB_{\cF}>0$ such that $\|m\|_{\infty}\leq B_{\cM}$ for all $m\in\cM$ and $\|f\|_{\infty}\leq B_{\cF}$ for all $f\in\cF$.


There are various notions available in the literature to quantify the complexity such as VC-dimension \citep{vapnik2015uniform}, metric entropy \citep{kolmogorov1959varepsilon}, pseudo-dimension \citep{pollard1990empirical}, fat-shattering dimension \citep{kearns1994efficient} etc. To cover a wide array of scenarios, we here quantify the complexity of $\cF$ through its Rademacher complexity, defined as follows.
\begin{definition}[Rademacher complexity]\label{def:rc}
Let $\cF$ be a function class and $\{x_i\}^n_{i=1}$ be a set of samples drawn i.i.d. from a distribution $\cD$. The Rademacher complexity of $\cF$ is defined as
\begin{align}
\label{rc}
 \cR_n(\cF) = \bbE\left[\sup_{f \in \cF}\frac1n \sum^n_{i=1} \eps_i f(x_i)\right], 
\end{align}
where $\{\eps_i\}^n_{i=1}$ are i.i.d. Rademacher random variables that equals to $\pm 1$ with probability $1/2$ each.
\end{definition}

Rademacher complexity provides insights into the function class's inherent complexity and its ability to generalize from the training data to unseen data.
This complexity can be bounded by leveraging the covering number, linking the abstract concept to the more tangible notion of covering the function space. 
The covering number, which quantifies the minimal number of balls of a given radius needed to cover the entire space, is defined as follows.
\begin{definition}[Covering number]
    \label{def:covering_number}
    Given a function class $\cF$ with some (semi)-metric $d$ on $\cF$, the $\eps$-covering number of $\cF$ is defined as: 
    $$
    \cN(\eps, d, \cF) = \inf\left\{m \ge 0: \exists f_1, \dots, f_m \in \cF \ \ \mathrm{ such \ that } \ \ \sup_{f \in \cF} \min_{1 \le j \le m} d(f_j, f) \le \eps\right\} \,.
    $$
The $\eps$-metric entropy of $\cF$ is defined as $\log{\cN(\eps, d, \cF)}$. 
\end{definition}
Depending on the growth rate of $N(\eps, \cF, d)$ with respect to $\eps$, we have two broad categories: 
\begin{enumerate}
    \item {\bf Polynomial growth: }$\cN(\eps, d, \cF)$ grows polynomially as $\eps \downarrow 0$, i.e., for some $\gamma>0$:
    \begin{equation}
    \label{eq:poly_growth}
    \cN(\eps, d, \cF) \lesssim \left(\frac{1}{\eps}\right)^\gamma \,.
    \end{equation}

    \item {\bf Exponential growth: }$\cN(\eps, d, \cF)$ grows exponentially as $\eps \downarrow 0$, i.e., for some $\gamma>0$:
    \begin{equation}
    \label{eq:exp_growth}
    \log{\cN(\eps, d, \cF)} \lesssim \left(\frac{1}{\eps}\right)^\gamma \,.
    \end{equation}
\end{enumerate}
As it is evident a function class with polynomial growth is \emph{less complex} than a function class with exponential growth. 
Within the subset of functions characterized by polynomial growth, the Vapnik-Chervonenkis (VC) class stands out for its unique properties. The concept of shattering, central to the VC class, involves the capacity of a collection of sets $\cC$ to distinguish among all subsets of a given set of points ${x_1,\ldots, x_n}$. Specifically, $\cC$ shatters this set if for every combination of points, there exists a subset in $\cC$ that intersects with the set of points in exactly that combination, making $|\{\{x_1,\ldots, x_n\}\cap C\mid C\in\cC\}|=2^{n}$. The VC dimension, denoted by $\VC(\cC)$, represents the largest number of points that can be shattered by $\cC$.
Correspondingly, the VC dimension of a function class is defined as follows:

\begin{definition}[Vapnik-Chervonenkis(VC) class]
A collection of measurable functions $\cF: \cX \mapsto \reals$ is deemed a VC class if the subgraphs of functions in $\cF$ form a VC class of sets. The VC dimension of $\cF$, denoted by $V_\cF$, is defined as the VC dimension of the collection of all subgraphs.
\end{definition}

For a bounded class of functions $\cF$ with finite VC dimension $V_\cF$,  
we can bound its covering number by Haussler's bound (see Theorem 2.6.7 of \cite{van1996weak}): 
\begin{align}\label{eq:Haussler_bound}
\textstyle
\sup_Q N(\eps B_{\cF}, \cF, L_2(Q)) & \le C V_\cF (16e)^{V_\cF}\left(\frac{1}{\eps}\right)^{2V_\cF - 1} \,
\end{align}
for some constants $C$. 


Finally, we capture the relation between the Rademacher complexity and the covering number via the following proposition:
\begin{proposition}
\label{prop:rc_cover}
Consider the covering number of $\cF$ in terms of $L_\infty$ norm. Then we have the following bound on the Rademacher complexity in terms of the growth of the covering number: 
$$
\cR_n(\cF)\leq 
\begin{cases}
C_1 n^{-\frac12}, & \text{ if the growth is polynomial or exponential with }0 < \gamma \leq 2\,, \\
C_2n^{-\frac{1}{\gamma}}, & \text{ if the growth is exponential with }\gamma > 2\,.
\end{cases}
$$
Here $C_1, C_2$ depend only on $\gamma$ and $B_{\cF}$. 
\end{proposition}
Thus, if  the covering number of $\cF$ grows polynomially or exponentially with $0 < \gamma < 2$, then the Rademacher complexity is of the order $n^{-1/2}$ (this is also known as \emph{Donsker class} of functions). Specifically, if $\cF$ is a VC class with VC dimension $V_{\cF}$, we then have $\cR_n(\cF)\leq C\sqrt{V_{\cF}/n}$ for some absoulte constant $C>0$. The Rademacher complexity increases with $\gamma$ when the growth of the covering number is  exponential with $\gamma> 2$.

\subsection{Analysis of initial estimation when $m_0$ is known}
\label{subsec:known_mean}
In this subsection, we present guarantees for the prediction interval (in terms of average width and coverage) assuming the conditional mean $m_0$ is known. Let us recall that when $m_0$ is known, then the optimization problem for estimating $\hat f_\pre$ (equation \eqref{sample_opt:general_form}) simplifies to: 
\begin{equation}
\label{sample_opt:known_mean}
    \begin{aligned}
        \min_{f \in \cF} & \ \ \ \ \frac1n \sum^n_{i=1}f(x_i) \\
        \st & \  \ \ \  f(x_i)\geq(y_i - m_0(x_i))^2 \ \ \ \forall \ 1 \le i \le n \,.
    \end{aligned}
\end{equation}
It is now natural to ask what function $\hat f_\pre$ approximates. This can be answered by looking at the population version of the optimization problem \eqref{sample_opt:known_mean}, which is the following:  
\begin{equation}
\label{population_opt:known_mean}
    \begin{aligned}
        \min_{f \in \cF} & \ \ \ \ \bbE[f(X)] \ \ \ \ \ \ \st \ \ f(X) \geq (Y - m_0(X))^2 \ \ \text{almost surely}. 
    \end{aligned}
\end{equation}
Let us define $f_0$ to be the minimizer of the above optimization problem when $\cF$ is the set of all functions. 
As $\hat f_\pre$ covers the sample responses, it is intuitive that $\hat f_\pre$ approximates $f_0$ if $f_0 \in \cF$.   
If $f_0 \notin \cF$, then the model mis-specification error, denoted by $\Delta(\cF)$, is defined as follows: 
\begin{equation}
\label{def:Delta_F}
\Delta(\cF) = \min_{\substack{f \in \cF \\ f \ge f_0 }} \bbE[f(X)-f_0(X)] \,.
\end{equation}
The following theorem presents a bound on the difference between the expected width of the estimated prediction interval and its corresponding population minimizer in terms of $\Delta(\cF)$ and the complexity of the underlying function class.

\begin{theorem}
\label{thm:known_mean}
Let $\hat f_{\pre}$ be the solution of \eqref{sample_opt:known_mean}. Then for any $t \ge 0$, with probability at least $1 - 2e^{-t}$ the following bound holds: 
\begin{align}
\label{ineq:known_mean}
\bbE[\hat f_\pre (X)\mid \cD] - \bbE[f_0(X)] \leq \Delta(\cF) + 4 \cR_n(\cF) + 4B_{\cF}\sqrt{\frac{t}{2n}}\,.
\end{align}
Moreover, if $\cF$ is a VC class, then with probability at least $1 - e^{-t}$ the following coverage guarantee holds:
\begin{align}
\label{cov:vc_known_mean}
 \bbP\left((Y - m_0(X))^2 \le \hat f_\pre(X) \,\big|\,\cD\right)\geq 1-c\sqrt{\frac{\VC(\cF)}{n}}-c\sqrt{\frac{t}{n}}\,
\end{align}
for some positive constant $c$. If $\cF$ is a non-VC class, then with probability at least $1 - e^{-t}$ the following coverage guarantee holds for any $\delta>0$. 
\begin{align}\label{cov:known_mean}
\bbP\left((Y - m_0(X))^2 \le \hat f_\pre(X) + \delta \,\big|\,\cD\right)   \geq 1-\frac{2}{\delta}\cdot\cR_n(\cF)-\sqrt{\frac{2t}{n}}.
\end{align}
\end{theorem}

\begin{remark}\label{rmk:cov}
    In Theorem \ref{thm:known_mean}, for non-VC class, we have obtained a coverage guarantee in terms of $\cR_n(\cF)$ and $\delta$. Suppose the covering number of $\cF$ grows exponentially with some $\alpha>0$. For $0 < \alpha \leq 2$, we know that $\cR_n(\cF)\leq\cO(n^{-1/2})$ and all we need is $\delta \gg n^{-1/2}$ for its contribution to be negligible, whereas for $\alpha > 2$, we know that $\cR_n(\cF)\leq \cO(n^{-1/\alpha})$ and $\delta \gg n^{-1/\alpha}$ is needed. For the former case, if we choose $t = c_1\log{n}$ and $\delta = c_2 (\log{n})^{-1/2}$ for some positive constants $c_1$ and $c_2$, we have 
    $$
    \bbP\left((Y - m_0(X))^2 \le \hat f_\pre(X) \,\big|\,\cD\right) \ge 1 - c_3\sqrt{\frac{\log{n}}{n}} \ \ \text{w.p. greater than } 1 - n^{-c_1} \,, 
    $$
    for a positive constant $c_3$.
    Similar choices of $(t, \delta)$ for $\alpha > 2$ yields
     $$
    \bbP\left((Y - m_0(X))^2 \le \hat f_\pre(X) + \delta \,\big|\,\cD\right)  \ge 1 - c_3\frac{\sqrt{\log{n}}}{n^{\frac{1}{\alpha}}} \ \ \text{w.p. greater than } 1 - n^{-c_1} \, .
    $$   
\end{remark}

\subsection{Analysis of initial estimation when $m_0$ is unknown}
\label{subsec:unknown_mean}
In this section, we present our analysis for constructing the prediction interval using either i) one-step estimation (Algorithm \ref{algo:main_algo}) which learns $(\hat f_\pre, \hat m)$ simultaneously by jointly optimizing \eqref{sample_opt:general_form} or ii) two-step estimation (Algorithm \ref{algo:main_algo_twostep}), where we first estimate $m_0$ via regression and then plugin that estimator in \eqref{sample_opt:general_form} to get $\hat f_\pre$. 

As in the previous subsection, the natural question is what population quantity $\hat f_\pre$ approximates. Here the situation is somewhat complicated as the estimation also depends on $m$. We define a quantity $f_m$ as the minimizer of the following optimization problem when $m_0$ is replaced by some $m \in \cM$: 
\begin{equation}
\label{population_opt:unknown_mean_1}
    \begin{aligned}
        \min_{f} & \ \ \ \ \bbE[f(X)] \ \ \ \ \ \ \st \ \ f(X) \geq (Y - m(X))^2 \ \ \text{almost surely}. 
    \end{aligned}
\end{equation}
In other words, $f_m$ is an almost sure upper bound on $(Y - m(X))^2$ for $m \in \cM$. 
Now the $\min_{m \in \cM} \bbE[f_m(X)]$ gives the minimal width when $m$ varies over $\cM$ and it is rational to imagine that the average width of $\hat f_\pre(X)$ approximates $\min_{m \in \cM} \bbE[f_m(X)]$. We divide the analysis for our method into two subsections; \ref{sec:one_step_theory} studies Algorithm \ref{algo:main_algo} and subsection \ref{sec:two_step_theory} analyzes Algorithm \ref{algo:main_algo_twostep}.  

\subsubsection{One-step estimation}\label{sec:one_step_theory}

In the one step procedure, we obtain $(\hat f_\pre, \hat m)$ by jointly optimizing \eqref{sample_opt:general_form} over $(f, m)$.
Analogous to \eqref{def:Delta_F}, we define $\Delta(\cM, \cF)$ as: 
\begin{equation}
    \label{def:delta_M_F}
    \Delta(\cM,\cF) = \min_{\substack{m\in\cM, f \in \cF \\ f \geq f_m}} \bbE[f(X)-f_0(X)] \,.
\end{equation}
Furthermore, as evident from the constraint of \eqref{sample_opt:general_form}, the complexity of $(\cM, \cF)$ appears through the function $f(X)- (Y - m(X))^2$. Therefore, we define a function class $\cG$ as:  
\begin{equation}
\label{eq:G_def}
\cG = \left\{g(x, y) = f(x)-(y - m(x))^2: f \in \cF, m \in \cM\right\} \,.
\end{equation}
Our next theorem provides a bound for the width of the estimated prediction interval obtained via Algorithm \ref{algo:main_algo} as well as a coverage guarantee in terms of the model mis-specification error and complexity.
\begin{theorem}
\label{thm:one_step}
For any $t\ge 0$, with probability at least $ 1 - 2e^{-t}$ the following bound holds: 
\begin{align}
\label{ineq:general_form}
\bbE[\hat f_\pre(X) \mid \cD] - \bbE[f_0(X)] \leq \Delta(\cM,\cF) + 4 \cR_n(\cF) + 4{B_{\cF}}\sqrt{\frac{t}{2n}}\,.
\end{align}
Moreover, if $\cG$ is a VC class, then with probability at least $1 - e^{-t}$ the following coverage guarantee holds:
\begin{align}\label{ineq:vc_coverage}
 \bbP\left((Y - \hat m(X))^2 \le \hat f_\pre(X) \,\big|\,\cD\right)  \geq 1-c\sqrt{\frac{\VC(\cG)}{n}}-c\sqrt{\frac{t}{n}}\,
\end{align}
for some positive constant $c$. If $\cG$ is a non-VC class, then with probability at least $1 - e^{-t}$ the following coverage guarantee holds for any $\delta>0$:
\begin{align}
\label{ineq:coverage}
\bbP\left((Y - \hat m(X))^2 \le \hat f_\pre(X) + \delta \,\big|\,\cD\right)\geq 1-  \frac{2}{\delta}\cdot\cR_n(\cG)-\sqrt{\frac{2t}{n}}.
\end{align}
\end{theorem}
Note that, the upper bound \eqref{ineq:general_form} incorporates the complexity of $\cM$ only through the bias term $\Delta(\cM, \cF)$. Therefore, if $\cM$ is highly complex (say a set of all measurable functions), then $\Delta(\cM, \cF)$ is small, and the bound on \eqref{ineq:general_form} improves. However, we pay the price in the coverage, which explicitly depends on the complexity of $\cG$. Here the complexity of $\cG$ is typically the maximum of complexity of $\cM$ and $\cF$, i.e. $\cR_n(\cG)=\cO(\max\{\cR_n(\cM), \cR_n(\cF)\})$. To understand this trade-off, consider a simple example where $\cM$ is a collection of all measurable functions. Then we can always choose $\hat m$ such that $y_i = \hat m(x_i)$ for all $1 \le i \le n$ (note that a $n$ degree polynomial is sufficient to perfectly fit the data). Consequently, for any $\cF$, we choose $f \in \cF$ which has the smallest empirical mean (it is $0$ if $0 \in \cF$), without any regard to the prediction interval leading to poor coverage. That is why we need to control the complexity of $\cG$ for adequate coverage as illustrated in \eqref{ineq:vc_coverage} and \eqref{ineq:coverage}. The choice of $\delta$ aligns with the discussion in Remark \ref{rmk:cov}. 

\subsubsection{Two-step procedure}\label{unk:two-step}
\label{sec:two_step_theory}
In this subsection, we analyze the two-step procedure as outlined in Algorithm \ref{algo:main_algo_twostep}. 
As the one-step procedure optimizes over two function classes $\cM$ and $\cF$ simultaneously, it can be computationally challenging. For example, if $\cM$ is a class of smooth $k$ times differentiable functions (e.g. Hölder class, Sobolev class, or Besov class of functions) then it is not immediate how to optimize $m$ and $f$ efficiently in one shot. To overcome this issue, we resort to Algorithm \ref{algo:main_algo_twostep}, where we first learn $\hat m$ and then solve \eqref{sample_opt:general_form} to obtain $\hat f_\pre$ with $\hat m$ fixed, i.e.: 
\begin{equation}\label{sample_opt:fix_m}
    \begin{aligned}
        \min_{f \in \cF} & \ \ \ \ \frac1n \sum^n_{i=1}f(x_i) \\
        \st & \  \ \ \  f(x_i)\geq(y_i - \hat m(x_i))^2 \ \ \ \forall \ 1 \le i \le n \,.
    \end{aligned}
\end{equation}
Since $\hat m$ is fixed and we only optimize over the function class $\cF$, this two-step procedure is more computationally efficient compared to the one-step procedure.

Suppose that $\hat m$ is a consistent estimator with respect to $L_\infty$ norm, i.e. $\|\hat m(x)-m_0(x)\|_{\infty}$ is small with high probability. One may expect that the optimal value of \eqref{sample_opt:fix_m} is close to the optimal value of \eqref{sample_opt:known_mean}, as the constraints in both problems are close to each other. However, this is not always the case as it heavily relies on $\cF$. Consider a toy example where $\cF=\{f_0(x),f_0(x) + 1\}$ and $\hat m(x)=m_0(x)+\varepsilon$. Here $\varepsilon$ is an arbitrarily small constant. The optimal solution for \eqref{sample_opt:fix_m} is $f_0(x) + 1$ since $f_0(x)$ does not satisfy the constraint, whereas the optimal solution for \eqref{sample_opt:known_mean} is $f_0(x)$. Hence the difference between the two optimal values is $1$, which does not depend on $\varepsilon$ due to the lack of richness of the class $\cF$ for continuity. In other words, if $\cF$ does not admit some sort of \emph{continuity} with respect to the perturbation of the constraints, then it is not possible to quantify the closeness of the optimal value of \eqref{sample_opt:fix_m} and \eqref{sample_opt:known_mean} with respect to the closeness between $\hat m$ and $m_0$. We now state our continuity assumption; define $f^\star$ to be the \emph{best approximator} of $f_0$ in $\cF$, i.e. $f^{\star}=\argmin_{f\in\cF:f\geq f_0 }\bbE[f(X)]$. Our continuity assumption requires that $f^{\star}$ should not be an isolated point, i.e. given $c>0$, $f^{\star} + c \in \cF$ for all \emph{small} $c$. In other words, our assumption brings certain continuity to our optimization problem with respect to the perturbation of $m_0$. We now formally state the assumption: 
\begin{assumption}\label{assm}
Let $f^{\star}=\argmin_{f\in\cF:f\geq f_0}\bbE[f(X)]$. There exists $D>0$ such that for any $c\in [0,D]$, we have $f^{\star}+c\in\cF$.
\end{assumption}
Based on the above assumption, we now present a theoretical guarantee for the two-step estimation below: 
\begin{theorem}
\label{thm:two_step}
Let $\hat f_\pre$ is obtained by solving \eqref{sample_opt:fix_m}. Furthermore, assume that Assumption \ref{assm} holds and $\|\hat m-m_0\|_{\infty}\leq c_1D$. Then for any $t>0$, with probability at least $1-2e^{-t}$, we have 
$$
\bbE[\hat f_\pre(X)\mid \cD] -\bbE[f_0(X)]\leq\Delta(\cF)+c_2\|\hat m-m_0\|_{\infty}+4\cR_n(\cF)+4B_{\cF}\sqrt{\frac{t}{2n}},
$$
where $c_1$ and $c_2$ are constants depending on the boundness of $Y$ and $\cM$. Moreover, if $\cF$ is a VC class, then with probability at least $1 - e^{-t}$ the following coverage guarantee holds:
\begin{align}\label{cov:vc_two_step}
 \bbP\left((Y - \hat m(X))^2 \le \hat f_\pre(X) \,\big|\,\cD\right)  \geq 1-c\sqrt{\frac{\VC(\cF)}{n}}-c\sqrt{\frac{t}{n}}\,
\end{align}
for some constant $c$. If $\cF$ is a non-VC class, then with probability at least $1 - e^{-t}$ the following coverage guarantee holds:
\begin{align}\label{cov:two_step}
 \bbP\left((Y - \hat m(X))^2 \le \hat f_\pre(X) + \delta \,\big|\,\cD\right)  \geq 1-\frac{2}{\delta}\cdot\cR_n(\cF)-\sqrt{\frac{2t}{n}}\,
\end{align}
for any $\delta>0$.
\end{theorem}

We now compare our theoretical findings for the one-step and the two-step algorithm. In terms of the bound on the excess width, the key difference between Theorem \ref{thm:one_step} and Theorem \ref{thm:two_step} is in the bias term; the bias term $\Delta(\cM, \cF)$ in Theorem \ref{thm:one_step} is replaced by $\Delta(\cF) + c_2\|\hat m - m_0\|_\infty$ in Theorem \ref{thm:two_step}. As $\hat m \in \cM$ by definition, it is immediate that the bias term of Theorem \ref{thm:two_step} is larger than that of Theorem \ref{thm:one_step} and consequently Algorithm \ref{algo:main_algo} provides narrower width. This difference is negligible if $\|\hat m - m_0\|_\infty \to 0$ and $m_0 \in \cM$.

On the other hand, the two-step estimation is more computationally efficient than one-step since estimating $m_0$ using regression techniques is easy to implement. Furthermore, the bias term in Theorem \ref{thm:two_step} is more interpretable as it precisely quantifies how an estimation error in $m_0$ propagates through the estimation error in the average width of the prediction intervals. Furthermore, establishing guarantees for $\hat m$ obtained through two-step estimation is also easier compared to that of one-step, as $\hat m$ depends on both function classes $\cM$ and $\cF$ in one-step estimation.

\subsection{Coverage guarantee for UTOPIA}
\label{sec:coverage_guarantee}
To achieve a user-specific coverage level $\alpha > 0$, we re-scale $\hat f_\pre$ by a factor of $\hat \lambda(\alpha)$ as outlined in Algorithms \ref{algo:main_algo} and \ref{algo:main_algo_twostep}. For the purpose of theoretical analysis, we split $n$ data into two parts: (1) $\cD_1$: $n_1$ samples that are used to obtain the initial estimator $\hat f_\pre$ and $\hat m$; (2) $\cD_2$: $n_2=n-n_1$ samples that are used to learn the shrink level $\hat \lambda(\alpha)$. 
More specifically, we first use $\cD_1$ to run Step 1 of Algorithm \ref{algo:main_algo} or Step 1 and Step 2 of Algorithm \ref{algo:main_algo_twostep}. 
We then \emph{remove some outliers} from $\cD_2$, in the sense we further select a subset $\cI \subseteq \cD_2$ defined as: 
$$
\cI:=\left\{n_1 + 1\leq i\leq n: (y_i-\hat m(x_i))^2\leq \hat f_\pre(x_i)+\delta\right\} \,.
$$
Finally, we estimate $\hat \lambda(\alpha)$ based on this selected sub-samples in $\cI$, i.e.
$$
\hat\lambda({\alpha}):=\inf_{\lambda}\left\{\lambda: \frac{1}{|\cI|} \sum_{i\in\cI}\mathbf{1}_{\{(y_i-\hat m(x_i))^2>\lambda(\hat f_\pre(x_i)+\delta)\}}\leq \alpha\right\} \,.
$$
The choice of $\cI$ is to avoid some \emph{extreme observations} in $\cD_2$; for example if $(y_i - \hat m(x_i))^2$ is much larger than $\hat f_\pre(x_i) + \delta$ for some of the observations in $\cD_2$, then choosing $\hat\lambda({\alpha})$ based on the entire samples $\cD_2$ may become unstable yielding a very large $\hat \lambda(\alpha)$, whereas intuitively we expect $0 \le \hat \lambda(\alpha) \le 1$. To avoid this pitfall, we restrict ourselves to the observations in $\cI$. In the previous sections, we have shown that with high probability $(Y - \hat m(X))^2 \le \hat f_{\pre}(X) + \delta$ and therefore $\cI$ is close to $\cD_2$.
\\\\
\noindent
Define $p_{n_1} $ to be the following conditional probability:
\begin{equation}
\label{eq:def_pn1}
p_{n_1} =  \bbP\left((Y - \hat m(X))^2 \le \hat f_\pre(X) + \delta \,\big|\,\cD_1\right)  \,.
\end{equation}
In the previous coverage guarantees and Remark \ref{rmk:cov}, we have discussed the choice of $\delta$ and quantified the rate of convergence of $p_{n_1} $ to $1$ in terms of the complexity of the corresponding function class. 
The subsequent theorem establishes the $(1-\alpha)$ level coverage guarantee of our final (shrinked) estimator: 
\begin{theorem}
\label{thm:adjust}
Under some weak continuity condition (i.e. $(Y - \hat m(X))/\hat f_\pre(X)$ has continuous distribution), with probability at least $1-2{n_2}^{-10 }$, we have
\begin{align*}
\left|\bbP\left(Y\in\widehat{\mathrm{PI}}_{(1-\alpha)}(X,\delta)\,\big|\,\cD\right)-(1-\alpha)\right|\leq 2(1-p_{n_1} )+c\cdot\sqrt{\frac{\log n_2}{n_2}} \,
\end{align*}
for some absolute constants $c$ and large enough $n_1,n_2$. Here $\widehat{\mathrm{PI}}_{(1-\alpha)}(X,\delta)$ follows the same definition as in Algorithm \ref{algo:main_algo} and \ref{algo:main_algo_twostep}.
\end{theorem}
The coverage guarantee of the preceding theorem relies on $p_{n_1}$, which can be further bounded via Theorem \ref{thm:one_step}. In particular, if both $\cM$ and $\cF$ are VC class of functions (e.g. for model aggregation type problems), then by Theorem \ref{thm:one_step}, with probability at least $1 - n^{-c_1}_1$, we have
\begin{align*}
\textstyle
p_{n_1} \geq 1-c\sqrt{\frac{\log n_1}{n_1}}\,.
\end{align*}
For non-VC class, choosing $\delta=c_2(\log n_1)^{-1/2}$ yields with probability greater than $1-n_1^{-c_1}$:  
$$
p_{n_1} \geq 1-c_3\frac{\sqrt{\log n_1}}{n_1^{\min({1/ 2}, {1/\alpha})}}. 
$$
Thus, for the Donsker class, with probability at least $1-2{n_2}^{-10}-n^{-c_1}_1$, we have
\begin{align*}
\left|\bbP\left(Y\in\widehat{\mathrm{PI}}_{(1-\alpha)}(X,\delta)\,\big|\,\cD\right)-(1-\alpha)\right|\le C\left(\sqrt{\frac{\log n_1}{n_1}}+\sqrt{\frac{\log n_2}{n_2}}\right) \,,
\end{align*}
for large enough $n_1,n_2$ and some constant $C > 0$. For the case $\alpha>2$, with probability at least $1-2 n^{-10}_2-n^{-c_1}_1$, we have
\begin{align*}
\left|\bbP\left(Y\in\widehat{\mathrm{PI}}_{(1-\alpha)}(X,\delta)\,\big|\,\cD\right)-(1-\alpha)\right|\le C\left(\frac{\sqrt{\log n_1}}{n^{\frac1\alpha}_1}+\sqrt{\frac{\log n_2}{n_2}}\right) \,.
\end{align*}
for large enough $n_1,n_2$ and some constant $C > 0$. 

\section{Applications}
\label{sec:applications}
In this section, we present several examples where our theory can be applied to obtain a prediction interval with a small width as well as a desired coverage guarantee.  

\subsection{Aggregation of Predictive Intervals}
\label{sec:model_aggregation}
In this section, we expand upon the example of the model aggregation presented in Section \ref{sec:aggregation}. 
We aim to obtain prediction interval of $Y$ (given $X$) of the form $[m(x) \pm \sqrt{f(x)}]$. Here we are given some estimates $\cM = \{m_1, \dots, m_L\}$ of the mean function and $\cF = \{f_1, \dots, f_K\}$ of the deviation from the mean, and we aggregate these estimates effectively to obtain a prediction interval with adequate coverage with minimal width. Towards that direction, we start with 
linear span of $\cF_0$ and $\cM_0$ defined as:
\begin{align}\label{agg:mf}
\cF :=\left\{f(x)=\sum^K_{j=1}\alpha_j f_j(x)\,\bigg|\, \alpha_j\geq 0 \right\},\ 
\cM:=\left\{m(x)=\sum^L_{j=1}\beta_j m_j(x)\,\bigg|\, \beta_j \in\mathbb{R}\right\}.
\end{align}
One may further truncate the functions in $\cM$ and $\cF$ at a desired level $B_{\cM}$ and $B_{\cF}$. We then optimize over the function classes $\cM$ and $\cF$ to find a suitable linear combination that obtains a small width and adequate coverage by solving \eqref{sample_opt:aggregation}. 
For establishing theoretical guarantees, note that $\cM$ and $\cF$ are linear vector spaces generated by the functions in $\cM_0$ and $\cF_0$ and consequently they are VC class of functions (as a linear subspace spanned by finitely many functions is VC with VC dimension less than or equal to the 2 + dimension of the subspace (see Lemma 2.6.15 of \cite{van1996weak})). Therefore, by Theorem \ref{thm:one_step}, we have the following Corollary of whose proof is deferred to Appendix: 
\begin{corollary}
\label{rc:aug}
We consider Algorithm \ref{algo:main_algo} with $(\cM,\cF)$ defined in \eqref{agg:mf}. With probability at least $1-e^{-t}$, the average squared width and the coverage of the estimated prediction interval \eqref{PI_aggregation} satisfies: 
\begin{align*}
    \bbE[\hat f_\pre(X) \mid \cD] - \bbE[f_0(X)] & \leq \Delta(\cM,\cF) + c \sqrt{\frac{K}{n}} + 4B_{\cF}\sqrt{\frac{t}{2n}} \\    \bbP\left((Y - \hat m(X))^2 \le \hat f_\pre(X) \,\big|\,\cD\right)  & \geq 1-   c \sqrt{\frac{K + L^2}{n}}-c\sqrt{\frac{t}{n}} 
\end{align*}
The coverage guarantee for our final prediction interval $[\hat m(x) \pm \sqrt{\hat \lambda(\alpha)\hat f_\pre(x)}]$ is also immediate from Theorem \ref{thm:adjust}.
\end{corollary}
Although the above corollary is based on the one-step method, the analysis of the two-step method follows similarly. First, we estimate $m_0$ via finding the optimal linear combination of the functions of $\cM$, and then we use that $\hat m$ to estimate $f_0$ by solving a linear programming problem. Assumption \ref{assm} is satisfied if we include a constant estimator in our candidate estimators. We skip the details here for the sake of brevity.

\subsection{Nonparametric Construction of Predictive Intervals}
In this application, we consider the general non-parametric setting, i.e. both $m_0$ and $f_0$ (see \eqref{population_opt:known_mean} for the definition) are \emph{smooth} functions, say they belong to a Hölder class or a Sobolev class. Suppose $\cX\times\cY\subseteq [0,1]^{d+1}$. Recall that, a Hölder class of functions from $[0,1]^d$ to $\reals$, denoted by $C_L^k([0, 1]^d)$, is defined as: 
\begin{align}
C_L^k([0, 1]^d) & = \left\{f: [0, 1]^d \to \reals: \max_{\substack{\bl = (l_1, \dots, l_d) \\ \sum_i l_i \le \lfloor k \rfloor}}\sup_x \left|\frac{\partial^{\bl}}{\partial x_1^{l_1} \dots \partial x_d^{l_d}} f(x)\right| \right. \notag \\
\label{eq:def_holder_class} & \hspace{5em} \left. + \max_{\substack{\bl = (l_1, \dots, l_d) \\ \sum_i l_i = \lfloor k \rfloor}}\sup_x \frac{\left|\frac{\partial^{\bl}}{\partial x_1^{l_1} \dots \partial x_d^{l_d}} f(x) - \frac{\partial^{\bl}}{\partial x_1^{l_1} \dots \partial x_d^{l_d}} f(y)\right|}{\|x - y \|^{k - \lfloor k \rfloor}} \le L \right\}
\end{align}
Assume $\cF = C_F^\alpha([0, 1]^d)$ and $\cM = C_M^\beta([0, 1]^d)$ and the function classes are well-specified to avoid potential biases, i.e. $f_0 \in \cF$ and $m_0 \in \cM$. Then, it is immediately that $\Delta(\cM,\cF) = 0$. For one-step procedure, we solve \eqref{sample_opt:general_form} by optimizing over $\cM = C_M^\beta([0, 1]^d)$ and $\cF = C_F^\alpha([0, 1]^d)$ simultaneously. Let $\hat m$ and $\hat f_{\pre}$ be the solution.
Theorem \ref{thm:one_step} yields the following Corollary of whose proof is deferred to Appendix.
\begin{corollary}\label{cor:smooth_func}
We consider Algorithm \ref{algo:main_algo} with $\cF = C_F^\alpha([0, 1]^d)$ and $\cM = C_M^\beta([0, 1]^d)$. Suppose that $f_0 \in \cF$ and $m_0 \in \cM$. With probability at least $1-e^{-t}$, the following holds:
\begin{align*}
\bbE[\hat f_\pre (X)\mid \cD] - \bbE[f_0(X)] \lesssim
\begin{cases}
    \sqrt{\frac{t}{n}} \,, & \text{ if } \alpha > d/2, \\
    \frac{1}{n^{\alpha/d}} + \sqrt{\frac{t}{n}} \,, & \text{ if } \alpha < d/2 \,.
\end{cases}  
\end{align*}
and 
\begin{align}\label{eq:pn_nonparam_1}
\bbP\left((Y - \hat m(X))^2 \le \hat f_\pre(X) + \delta \,\big|\,\cD\right)   & \geq 1-  \frac{c}{\delta\sqrt{n}}-\sqrt{\frac{2t}{n}} \hspace{0.5in} \text{if }0 < \gamma < 2 \,,\\
\label{eq:pn_nonparam_2}
\bbP\left((Y - \hat m(X))^2 \le \hat f_\pre(X) + \delta \,\big|\,\cD\right)   & \geq 1-  \frac{c}{\delta n^{\frac{1}{\gamma}}}-\sqrt{\frac{2t}{n}} \hspace{0.5in} \text{if }\gamma > 2 \,.
\end{align}
where $\gamma := \max\{d/\alpha ,(d+1)/\beta\}$ and $c$ is a positive constant. The coverage guarantee for our final prediction interval $[\hat m(x) \pm \sqrt{\hat \lambda(\alpha)\hat f_\pre(x)}]$ then follows directly from Theorem \ref{thm:adjust}.
\end{corollary}

For the two-step UTOPIA, the results are similar, with an additional factor of $\|\hat m - m_0\|_\infty$ in Theorem \ref{thm:two_step}, where $\hat m$ is first estimated by non-parametric regression. When $\cM$ is well specified, then by standard techniques of non-parametric regression, one may obtain $\hat m$ such that $\|\hat m - m_0\|_\infty = O_p\left((n/\log{n})^{-\frac{\beta}{2\beta + d}}\right)$.
For more details, see Theorem 4.3 of \cite{belloni2015some}. 
\subsection{Estimation via deep neural network}

In modern statistical applications, often a predictor is estimated via deep neural networks for the ease of implementation and its adequate approximation ability. Estimation of non-parametric mean function under various conditions via fully/sparsely connected neural networks and its optimality has recently been explored in several works including but not limited to \cite{kohler2016nonparametric, bauer2019deep, schmidt2020nonparametric, kohler2021rate, fan2022noise, bhattacharya2023deep}. In our problem, we may also estimate the prediction interval via deep neural networks. 
Consider the two-step UTOPIA, as outlined in Algorithm \ref{algo:main_algo_twostep}. 
We first estimate $m_0$ using the least-squares method. Subsequently, we solve the optimization problem described in \eqref{sample_opt:fix_m}, where $\cF$ comprises deep neural networks.
Setting up the notations, define $\cF$ as the following function class:
\begin{align*}
\textstyle
  \cF_{\mathrm{DNN}}(d, N, L, 1):=\{\text{DNNs with width $N$, depth $L$, input dim $d$ and output dim $1$}\}.   
\end{align*}
Here $\cF_{\DNN}$ is the collection of all fully connected neural networks with at most $L$ number of hidden layers and each layer has at most $N$ neurons. We further truncate the output of the neural network at level $B_{\cF}$. There are several possible choices of activation functions, i.e. ReLU, leaky-ReLU, sigmoid, tanh, etc. 
Here we consider piecewise linear activation functions.

By \cite{bartlett2019nearly} (Theorem 8), $\cF_{\mathrm{DNN}}(d, N, L, 1)$ with piecewise linear activation
functions (e.g. ReLU or leaky-ReLU) is a VC class whose VC dimension is of the order $WL\log{W}$, where $W$ is the total number of parameters of the neural network and $L$ is the total number of hidden layers. 
For fully connected neural networks, we then have $W \sim dN + N^2L$, and consequently VC dimension is $\sim N^2L^2 \log{NL}$ (if $d$ is of order $NL$). 
Thus, we do not need to extend the estimated prediction interval by $\delta$, i.e. we consider the prediction interval of the form $[\hat m(x) \pm \hat \lambda(\alpha)\sqrt{\hat f_\pre(x)}]$. Theorem \ref{thm:two_step} yields the following Corollary of whose proof is deferred to the Appendix:

\begin{corollary}
\label{cor:DNN} 
Assume that $f_0$ belongs to Hölder class as defined in \eqref{eq:def_holder_class} with smoothness $\alpha$, i.e. $f_0 \in \cF = C_F^{\alpha}([0, 1]^d)$. 
We consider the two-step method with $\cF$ being a set of DNNs with width $N$, depth $L$ and piecewise linear activation functions. With probability at least $1-e^{-t}$, the following holds:
\begin{align*}
\textstyle
&\bbE[\hat f_\pre (X)\mid \cD] - \bbE[f_0(X)]  \leq c\left(\frac{NL}{\log{N}\log{L}}\right)^{-\frac{2\alpha}{d}}  \\
& \hspace{10em} + c \frac{NL\sqrt{\log{NL}}}{\sqrt{n}} + c\|\hat m-m_0\|_{\infty}+4B_\cF\sqrt{\frac{t}{2n}},\\
&\bbP\left((Y - \hat m(X))^2 \le \hat f_\pre(X) \,\big|\,\cD\right)   \geq 1-c\sqrt{\frac{N^2L^2\log{NL}}{n}}-c\sqrt{\frac{t}{n}}.
\end{align*}
The coverage guarantee for our final prediction interval $[\hat m(x) \pm \hat \lambda(\alpha)\sqrt{\hat f_\pre(x)}]$ can be obtained directly by applying Theorem \ref{thm:adjust}. 
\end{corollary}
Choosing $(N, L)$ such that $NL \sim n^{\frac{d}{2(2\alpha + d)}}$, we obtain the rate $n^{-\frac{\alpha}{2\alpha + d}}$, the standard non-parametric rate (upto $\log$ factor). 
One advantage of using the neural network method is its ability to learn unknown low-dimension composition structures.  If the underlying function is a composition of low-dimensional functions, then the convergence rate is improved to the low-dimensional convergence rate \citep{kohler2021rate,schmidt2020nonparametric,fan2022noise}.

\subsection{Prediction via kernel basis}
\label{RKHS}
In this subsection, we elaborate on \eqref{eq:finite_SDP} from Section \ref{sec:rkhs} that uses a reproducing kernel Hilbert space to estimate $(m, f)$. 
The method in Section \ref{sec:rkhs} can further be extended to the unknown mean case. We here illustrate the one-step method via modifying \eqref{eq:finite_SDP} as follows: 
\begin{equation}
\begin{aligned}
\label{eq:liang_modified_finite_mean}
    \min_{B\in\bbS^{n\times n}, m \in \cM} \ \ & \tr(KBK) \\
    \st \ \ & \left \langle K_{i}, BK_{i}\right \rangle  \ge \left(y_i - m(x_i)\right)^2 \ \forall \ 1 \le i \le n \\
    & \tr(KB) \le r, \ \ B \succeq 0 \,.
\end{aligned}
\end{equation}
Here the class $\cM$ can be anything a practitioner would like to choose; it can be a suitable non-parametric class (Hölder, Sobolev, Besov, etc.) or a different RKHS with another kernel $\tilde K$ (which can be the same as $K$) or even aggregate various estimate of $m$. As an example, if the practitioner decides to model $m_0$ by an RKHS with kernel $\tilde K$, then we can modify \eqref{eq:liang_modified_finite_mean} as: 
\begin{equation}
\begin{aligned}
\textstyle
    \min_{B, \bgamma \in \reals^n} \ \ & \tr(KBK) \\
    \st \ \ & \left \langle K_{i}, BK_{i}\right \rangle  \ge \left(y_i - \sum_{j = 1}^n \tilde K(x_i, x_j) \gamma_j\right)^2 \ \forall \ 1 \le i \le n \\
    & \tr(KB) \le r, \ \ B \succeq 0 \,.
\end{aligned}
\end{equation}
Similarly, if $\cM = \{m_1, \dots, m_L\}$ is a collection of some estimates of $m_0$, then one may aggregate them: 
\begin{equation}
\begin{aligned}
\textstyle
    \min_{B, \bgamma \in \reals^L} \ \ & \tr(KBK) \\
    \st \ \ & \left \langle K_{i}, BK_{i}\right \rangle  \ge \left(Y_i - \sum_{j = 1}^L m_j(X_i) \gamma_j\right)^2 \ \forall \ 1 \le i \le n \\
    & \tr(KB) \le r, \ \ B \succeq 0 \,.
\end{aligned}
\end{equation}
Let $\hat B$ and $\hat m$ be the solution of \eqref{eq:liang_modified_finite_mean}. 
We then define $\hat f_\pre(x):=\langle K_x, \hat BK_x \rangle$ with $K_x\in\bbR^n$ and its $i^{th}$ element $K_{x,i}=K(x_i,x)$. The following corollary yields a theoretical guarantee for the coverage and width of $(\hat m, \hat f_\pre)$: 
\begin{corollary}
    \label{cor:rkhs}
We consider Algorithm \ref{algo:main_algo} with $\cF = \{\langle \Phi(x), \cA(\Phi)(x)\rangle_\cH: \|\cA\|_\star \le r, \cA\succeq 0\}$. Suppose that $\sup_x K(x, x) \le b$. Let $(\hat m, \hat f_\pre)$ be obtained by solving \eqref{eq:liang_modified_finite_mean}. Then we have with probability $\ge 1 - e^{-t}$: 
$$
\textstyle
\bbE[\hat f_\pre(X) \mid \cD] - \bbE[f_0(X)] \leq \Delta(\cM,\cF) + \frac{16r\sqrt{b}}{\sqrt{n}}\sqrt{\bbE\left[K(X, X)\right]}  + 4rb\sqrt{\frac{t}{2n}}\,.
$$
Furthermore, we have
$$
\bbP\left((Y - \hat m(X))^2 \le \hat f_\pre(X) + \delta \,\big|\,\cD\right)  \geq 1-  \frac{2}{\delta}\cdot\left(\cR_n(\cM) + \frac{16r\sqrt{b}}{\sqrt{n}}\sqrt{\bbE\left[K(X, X)\right]}\right)-\sqrt{\frac{2t}{n}} \,
$$
for any $\delta>0$. The coverage guarantee for our final prediction interval $[\hat m(x) \pm \hat \lambda(\alpha)\sqrt{\hat f_\pre(x)}]$ can be obtained from a simple application of Theorem \ref{thm:adjust}. 
\end{corollary}

\section{Simulation studies}
\label{simulation}
In this section, we present some simulation studies to underscore the applicability of UTOPIA for aggregating prediction intervals with a small average width and adequate coverage guarantee. We consider the following six estimates: 
\begin{itemize}
    \item [(1)]\textbf{Estimator 1}$(f_1)$: A neural network based estimator with depth=1, width=10 that estimates the 0.6 quantile function of $(Y-m_0(X))^2\mid X=x$.
    \item [(2)]\textbf{Estimator 2}$(f_2)$: A fully connected feed forward neural network with depth=2 and width=50 that estimates the 0.7 quantile function of $(Y-m_0(X))^2\mid X=x$.
    \item [(3)]\textbf{Estimator 3}$(f_3)$: A quantile regression forest estimating the 0.8 quantile function of $(Y-m_0(X))^2\mid X=x$.
    \item [(4)]\textbf{Estimator 4}$(f_4)$: A gradient boosting model estimating the 0.9 quantile function of $(Y-m_0(X))^2\mid X=x$.
    \item [(5)]\textbf{Estimator 5}$(f_5)$: 
    A random forest-based estimator of the conditional variances $\var(Y \mid X)$ with 1000 trees and maximum depth = 5. 
    \item [(6)]\textbf{Estimator 6}$(f_6)$: The constant function $1$.
\end{itemize}
Here the quantile estimators are obtained by minimizing the corresponding check loss. 
If the mean is known, then we solve \eqref{sample_opt:known_mean}, whereas for unknown mean we resort to the two-step estimation procedure; 
first, we estimate the conditional mean function $m_0$ using a three-layer neural network with 50 neurons in each hidden layer and then solve \eqref{sample_opt:fix_m}. 
For each of the simulation setup, we first generate $n = 3000$ i.i.d. samples $\{x_i,y_i\}^n_{i=1}$ from various data generating processes (described below) . 
We then use $n_{pre}=1000$ data points to estimate the component functions $\{f_i: 1 \le i \le 6\}$ and the conditional mean function $m_0$ (if unknown), $n_{opt}=500$ data points to aggregate the component functions by solving the optimization problem \eqref{sample_opt:known_mean} for the known mean case and \eqref{sample_opt:fix_m} for unknown mean case, and finally $n_{adj}=500$ data points to adjust the interval (Step 3 of Algorithm \ref{algo:main_algo_twostep}) to achieve $(1 - \alpha) = 0.95$ coverage. 
Last but not least, we use the rest $n_t=1000$ data points as the testing data. 
Additionally, we employ a no-split version of UTOPIA in which the same set of $1000$ data points is utilized to learn all the components, aggregate them, and shrink them to obtain the final prediction interval.
We compare UTOPIA with three baseline methods: (1) standard prediction band using linear quantile regression (\textbf{LQR}): we use linear quantile regression using \texttt{QuantileRegressor} function of the \texttt{sklearn} module in python.  
(2) split conformal prediction (\textbf{SplitCF}) (localized version of Algorithm 2 in \cite{lei2018distribution}): we estimate the conditional mean (as described above, with a three-layer neural network with 50 neurons at each layer) and the conditional variance function (same as $f_5$) and set the score function $S(x, y) = |y - \hat m(x)|/\hat \sigma(x)$. We then choose the cutoff based on this estimated score function and a separate hold-out sample to obtain $95\%$ confidence interval. 
(3) semi-definite programming (\textbf{SDP}) method with Gaussian kernel proposed by \cite{liang2021universal}. In the sequel, we describe our three different setups, including (1) a location-scale model with a known mean function, (2) a location-scale model with an unknown mean function, and (3) a non-location-scale model.  
\vspace{-0.1cm}
\subsection{Setup 1} 
\label{sec:setup_1}
In our first simulation setup, we consider a simple location-scale model with the known mean function $m_0$ (which is taken to be $0$ here for simplicity). More precisely we assume $Y_i = \sqrt{v_0(X_i)}\xi_i$ where $X_i \sim {\rm Unif}([-1,1])$, $\xi_i \indep X_i$, $\xi_i\sim {\rm Unif}([-1,1])$ and $v_0(x) = 1 + 25 x^4$.   The results are depicted in Figure~\ref{Simulation1}, Table~\ref{tab:Simulation1_split} and Table~\ref{tab:Simulation1_nosplit}. 

\vspace{-0.7cm}
\begin{figure}[H]
    \centering
    \begin{subfigure}{0.4\textwidth}
        \centering
        \caption{UTOPIA}
        \includegraphics[width=0.75\textwidth]{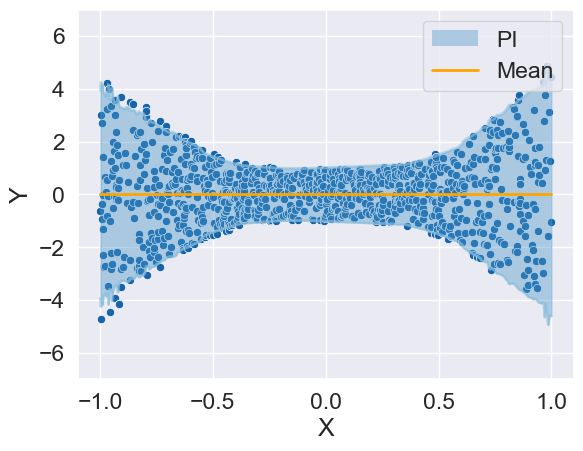}
        \label{Simulation1:our}
    \end{subfigure}
    \begin{subfigure}{0.4\textwidth}
        \centering
        \caption{LQR}
    
        \includegraphics[width=0.75\textwidth]{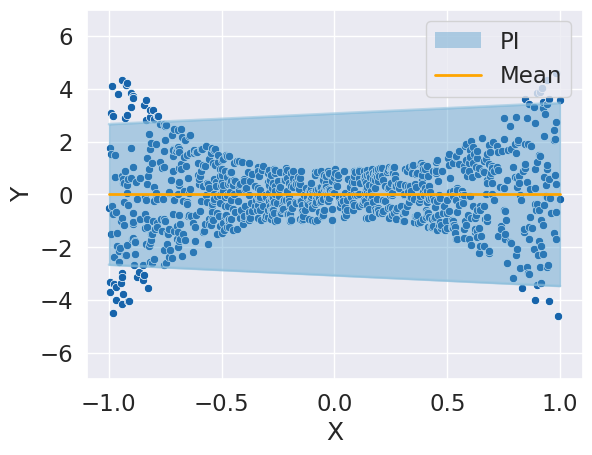}
        \label{Simulation1:LQR}
    \end{subfigure}
    \begin{subfigure}{0.4\textwidth}
        \centering
        \caption{SplitCF}
\includegraphics[width=0.75\textwidth]{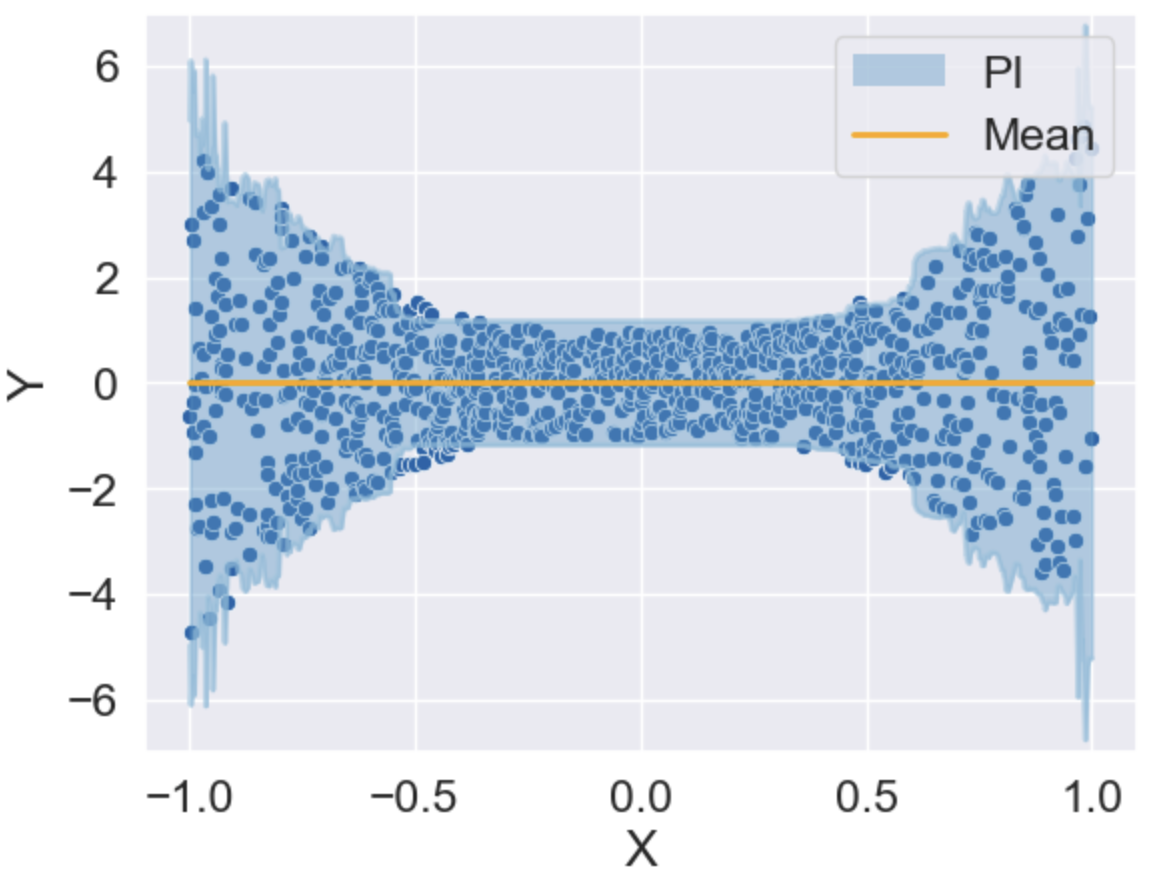}
\label{Simulation1:SplitCF}
    \end{subfigure}
    \begin{subfigure}{0.4\textwidth}
        \centering
        \caption{SDP}
     
        \includegraphics[width=0.75\textwidth]{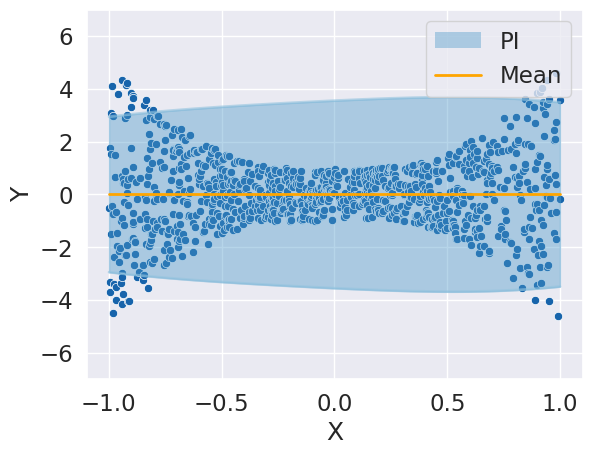}
        \label{Simulation1:SDP}
    \end{subfigure}
  \caption{The \textcolor{blue}{blue dots}, \textcolor{orange}{orange line} and \textcolor{NavyBlue}{light blue band} represent the testing data, the conditional mean function and the prediction band, respectively, for Setup 1.}
    \label{Simulation1}
\end{figure} 
\vspace{-1cm}
\begin{table}[H]
\centering
\begin{minipage}{0.9\textwidth}
\centering
\caption{Result for Setup 1 (over 200 Monte Carlo iterations) \b{with} splitting}
\label{tab:Simulation1_split}
\begin{tabular}{lcccc}
\toprule
& UTOPIA & LQR & SplitCF & SDP \\
\midrule
\b{Coverage} (Median) & $94.9\%$ & $94.8\%$ & $95\%$ & $96.9\%$ \\
\footnotesize Coverage (IQR) &\footnotesize ($0.015$) & \footnotesize ($0.015$) & \footnotesize ($0.012$) & \footnotesize ($0.014$) \\
\r{Width} (Median) & {\red $5.10$} & $9.69$ & $6.24$ & $14.32$\\
\footnotesize Width (IQR) & \footnotesize ($0.5$) & \footnotesize ($1.48$) & \footnotesize ($0.68$) & \footnotesize ($2.16$)\\
\bottomrule
\end{tabular}
\end{minipage}


\begin{minipage}{0.9\textwidth}
\centering
\caption{Result for Setup 1 (over 200 Monte Carlo iterations) \b{without} splitting}
\label{tab:Simulation1_nosplit}
\begin{tabular}{lcccc}
\toprule
 & UTOPIA & LQR & SplitCF & SDP \\
\midrule
\b{Coverage} (Median) & $94.2\%$ & $94.8\%$ & $94\%$ & $96.9\%$ \\
\footnotesize Coverage (IQR) & \footnotesize ($0.011$) & \footnotesize ($0.012$) & \footnotesize ($0.01$) & \footnotesize ($0.013$) \\
\r{Width} (Median) & {\red $4.97$} & $9.69$ & $5.83$ & $14.44$\\
\footnotesize Width (IQR) & \footnotesize ($0.27$) &\footnotesize ($1.02$) & \footnotesize ($0.42$) & \footnotesize ($2.01$)\\
\bottomrule
\end{tabular}
\end{minipage}
\end{table}

\subsection{Setup 2}
In our second setup for simulation, we extend our previous setup by adding an unknown mean to the location-scale model. More precisely here our model is $Y_i = m_0(X_i) + \sqrt{v_0(X_i)}\xi_i$, where as before $\xi_i\indep X_i$ and both are generated from ${\rm Unif}([-1,1])$. The true mean function $m_0(x) = 1 + 5x^3$ and $v_0(x) = 1 + 25x^4$.  The results are presented in Figure~\ref{Simulation2}, Table~\ref{tab:Simulation2_split} and Table~\ref{tab:Simulation2_nosplit}.

\begin{figure}[H]
    \centering
    \begin{subfigure}{0.4\textwidth}
        \centering
        \caption{UTOPIA}
  
        \includegraphics[width=0.75\textwidth]{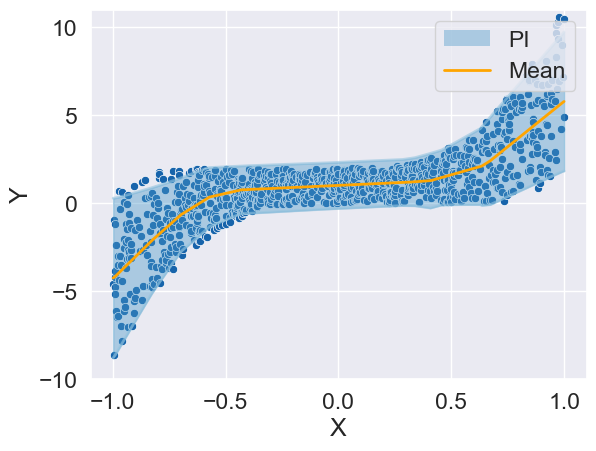}
        \label{Simulation2:our}
    \end{subfigure}
    \begin{subfigure}{0.4\textwidth}
        \centering
        \caption{LQR}
    
        \includegraphics[width=0.75\textwidth]{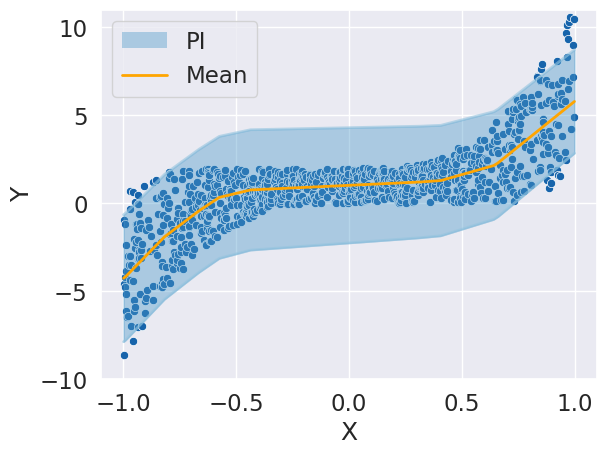}
        \label{Simulation2:LQR}
    \end{subfigure}
    \begin{subfigure}{0.4\textwidth}
        \centering
        \caption{SplitCF}
     
        \includegraphics[width=0.75\textwidth]{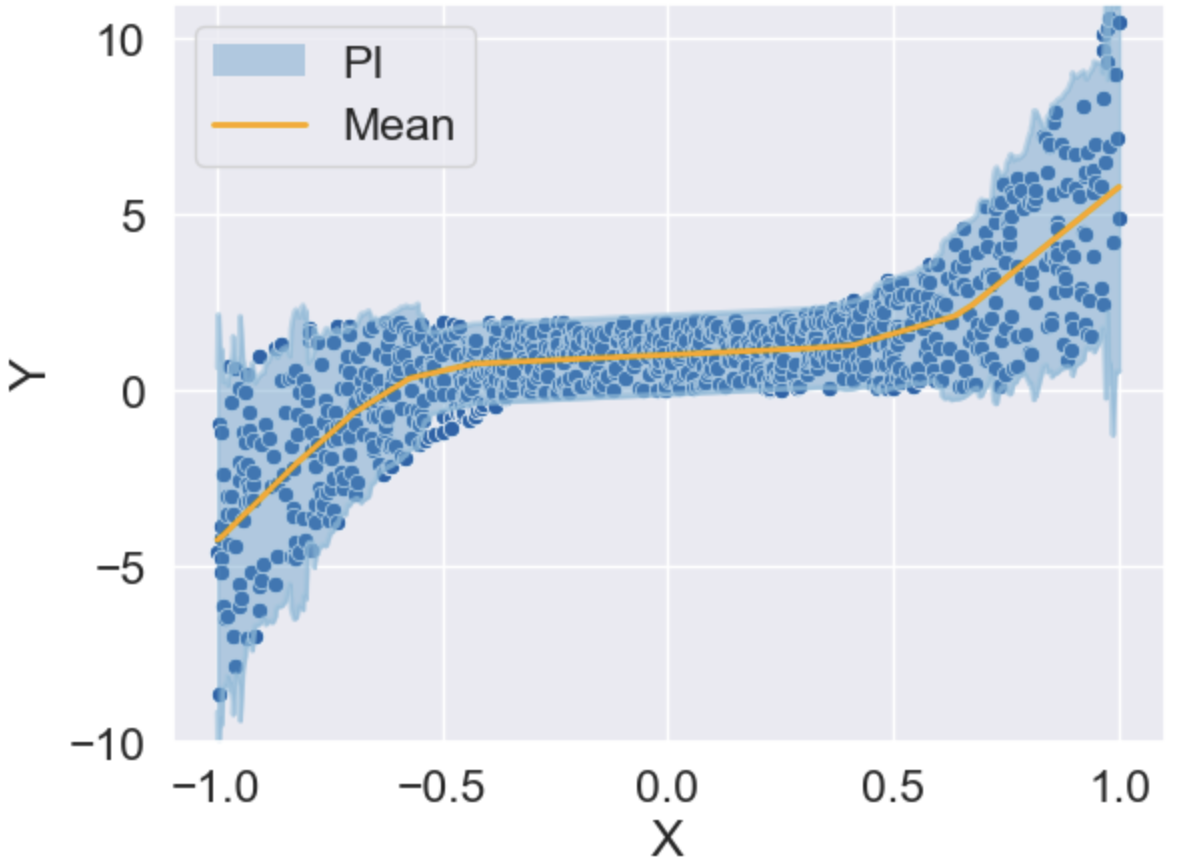}
        \label{Simulation2:SplitCF}
    \end{subfigure}
    \begin{subfigure}{0.4\textwidth}
        \centering
        \caption{SDP}
    
        \includegraphics[width=0.75\textwidth]{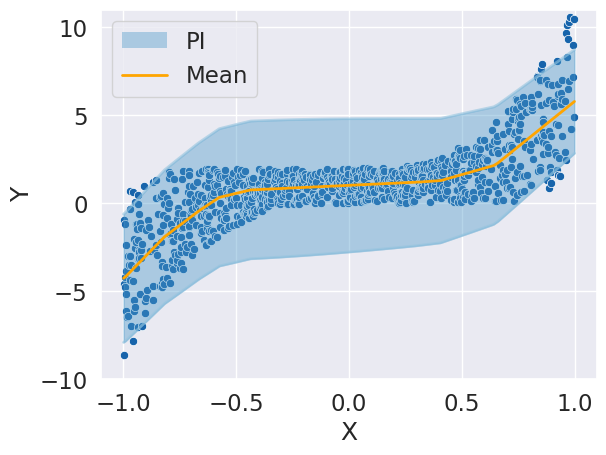}
        \label{Simulation2:SDP}
    \end{subfigure}
  \caption{The \textcolor{blue}{blue dots}, \textcolor{orange}{orange line} and \textcolor{NavyBlue}{light blue band} represent the testing data, the conditional mean function, and the prediction band, respectively, for Setup 2.}
    \label{Simulation2}
\end{figure}  

\vspace{-1cm}

\begin{table}[H]
\centering
\begin{minipage}{0.9\textwidth}
\centering
\caption{Result for Setup 2 (over 200 Monte Carlo iterations) \b{with} splitting}
\label{tab:Simulation2_split}
\begin{tabular}{lcccc}
\toprule
& UTOPIA & LQR & SplitCF & SDP \\
\midrule
\b{Coverage} (Median) & $95\%$ & $94.8\%$ & $95.1\%$ & $97\%$ \\
\footnotesize Coverage (IQR) &\footnotesize ($0.015$) & \footnotesize ($0.018$) & \footnotesize ($0.012$) & \footnotesize ($0.013$) \\
\r{Width} (Median) & {\red $5.35$} & $9.88$ & $6.38$ & $14.62$\\
\footnotesize Width (IQR) & \footnotesize ($0.59$) & \footnotesize ($1.67$) & \footnotesize ($0.65$) & \footnotesize ($2.67$)\\
\bottomrule
\end{tabular}
\end{minipage}

\vspace{.2cm}

\begin{minipage}{0.9\textwidth}
\centering
\caption{Result for Setup 2 (over 200 Monte Carlo iterations) \b{without} splitting}
\label{tab:Simulation2_nosplit}
\begin{tabular}{lcccc}
\toprule
& UTOPIA & LQR & SplitCF & SDP \\
\midrule
\b{Coverage} (Median) & $93.7\%$ & $94.8\%$ & $93.8\%$ & $96.9\%$ \\
\footnotesize Coverage (IQR) &\footnotesize ($0.013$) & \footnotesize ($0.014$) & \footnotesize ($0.013$) & \footnotesize ($0.013$) \\
\r{Width} (Median) & {\red $5.01$} & $9.68$ & $5.81$ & $14.19$\\
\footnotesize Width (IQR) & \footnotesize ($0.38$) & \footnotesize ($0.97$) & \footnotesize ($0.46$) & \footnotesize ($2.10$)\\
\bottomrule
\end{tabular}
\end{minipage}
\end{table}

\subsection{Setup 3} 
Our final setup is the most general setup, where we move beyond the location-scale model. We generate the inputs $X_i$ from ${\rm Unif}([-1,1])$ and the outputs $Y_i$ from a conditional Gamma distribution. 
More precisely, if $Z \sim \mathrm{Gamma}(k(X), \theta(X))$ with $k(x) =5+x, \theta(x) = 1+\sin x/2$, then $Y \sim Z|Z \leq 2k(X)\theta(X)$.  We summarize the results in Figure~\ref{Simulation3}, Table~\ref{tab:Simulation3_split} and Table~\ref{tab:Simulation3_nosplit}.
\begin{figure}[H]
    \centering
    \begin{subfigure}{0.4\textwidth}
        \centering
        \caption{UTOPIA}

    \includegraphics[width=0.75\textwidth]{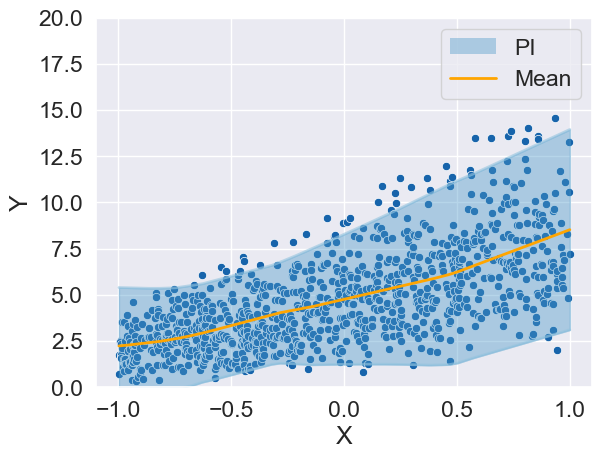}
        \label{Simulation3:our}
    \end{subfigure}
    \begin{subfigure}{0.4\textwidth}
        \centering
                \caption{LQR}

        \includegraphics[width=0.75\textwidth]{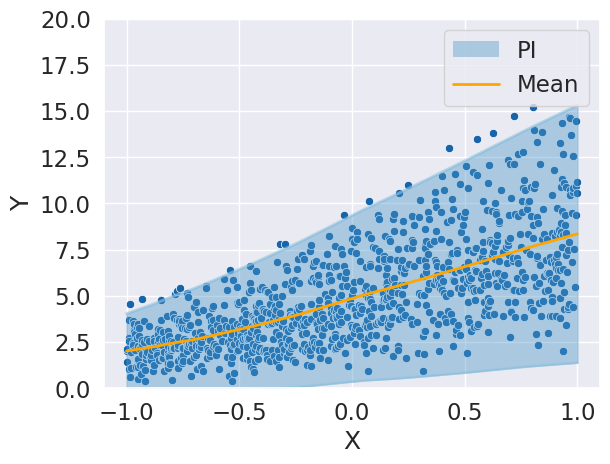}
        \label{Simulation3:LQR}
    \end{subfigure}
    \begin{subfigure}{0.4\textwidth}
        \centering
         \caption{SplitCF}

        \includegraphics[width=0.75\textwidth]{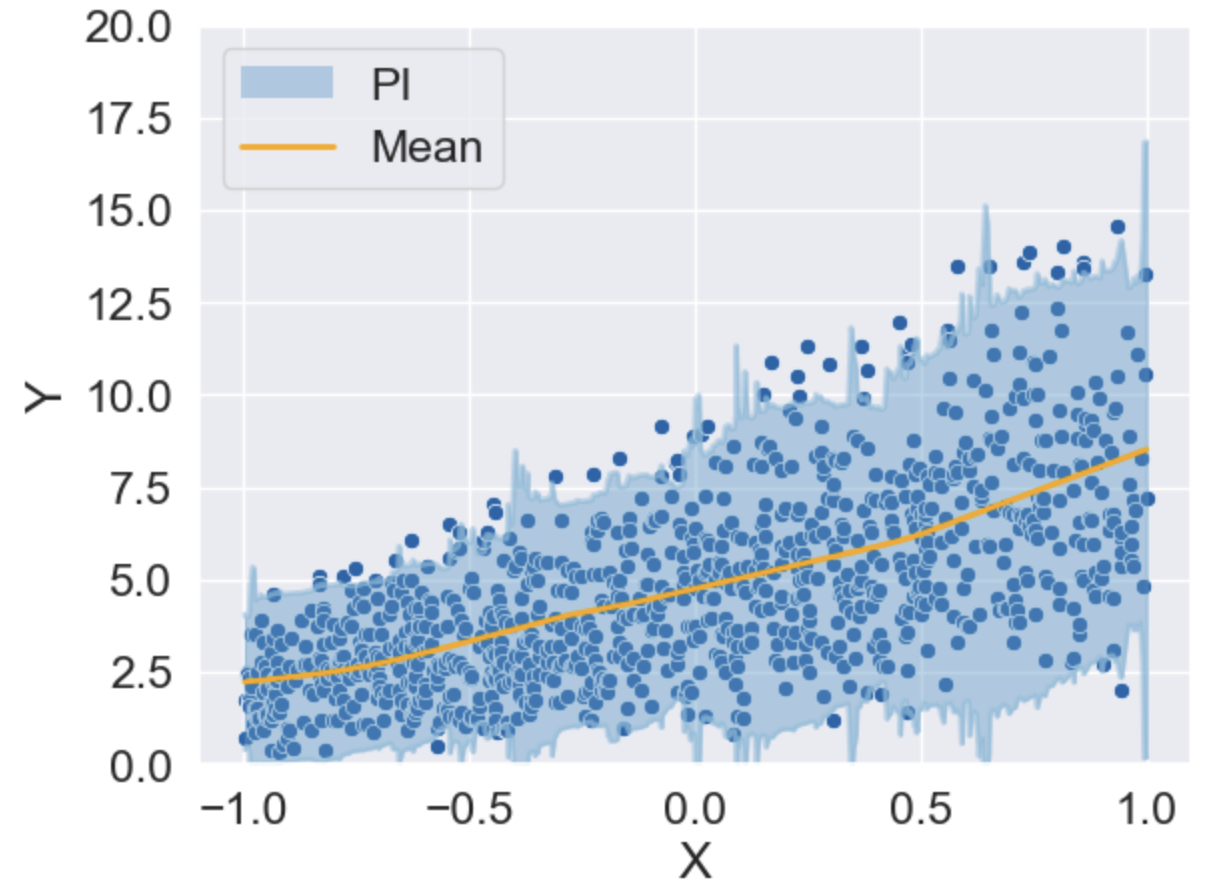}
        \label{Simulation3:SplitCF}
    \end{subfigure}
    \begin{subfigure}{0.4\textwidth}
        \centering
        \caption{SDP}

        \includegraphics[width=0.75\textwidth]{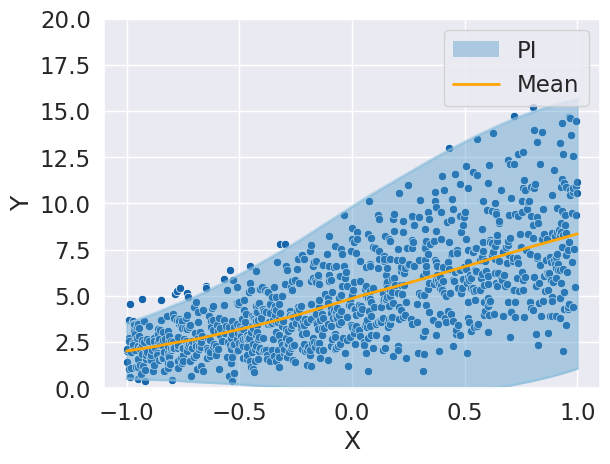}
        \label{Simulation3:SDP}
    \end{subfigure}
  \caption{The \textcolor{blue}{blue dots}, \textcolor{orange}{orange line} and \textcolor{NavyBlue}{light blue band} represent the testing data, the conditional mean function, and the prediction band, respectively, for Setup 3.}
    \label{Simulation3}
\end{figure}

\vspace{-1cm}

\begin{table}[H]
\centering
\begin{minipage}{0.9\textwidth}
\centering
\caption{Result for Setup 3 (over 200 Monte Carlo iterations) \b{with} splitting}
\label{tab:Simulation3_split}
\begin{tabular}{lcccc}
\toprule
& UTOPIA & LQR & SplitCF & SDP \\
\midrule
\b{Coverage} (Median) & $95\%$ & $97.4\%$ & $95.1\%$ & $96.9\%$ \\
\footnotesize Coverage (IQR) &\footnotesize ($0.017$) & \footnotesize ($0.011$) & \footnotesize ($0.015$) & \footnotesize ($0.015$) \\
\r{Width} (Median) & {\red $16.12$} & $20.73$ & $16.61$ & $19.93$\\
\footnotesize Width (IQR) & \footnotesize ($1.83$) & \footnotesize ($2.76$) & \footnotesize ($1.53$) & \footnotesize ($4.28$)\\
\bottomrule
\end{tabular}
\end{minipage}

\vspace{.2cm}

\begin{minipage}{0.9\textwidth}
\centering
\caption{Result for Setup 3 (over 200 Monte Carlo iterations) \b{without} splitting}
\label{tab:Simulation3_nosplit}
\begin{tabular}{lcccc}
\toprule
 & UTOPIA & LQR & SplitCF & SDP \\
\midrule
\b{Coverage} (Median) & $92.7\%$ & $97.4\%$ & $94.3\%$ & $96.9\%$ \\
\footnotesize Coverage (IQR) & \footnotesize ($0.015$) & \footnotesize ($0.008$) &\footnotesize ($0.011$) & \footnotesize ($0.014$) \\
\r{Width} (Median) & {\red $13.72$} & $20.86$ & $15.46$ & $20.25$\\
\footnotesize Width (IQR) & \footnotesize $(0.96$) & \footnotesize ($1.63$) & \footnotesize ($0.90$) & \footnotesize ($4.22$)\\
\bottomrule
\end{tabular}
\end{minipage}
\end{table}

All the methods yield coverage close to $95\%$, with results notably similar both with and without data splitting. In each configuration, UTOPIA delivers the lowest average width among the four approaches. This is due to the misspecification of the quantile in the cases of LQR, SplitCF, and SDP. In comparison, UTOPIA aggregates various estimators, enabling it to better capture the shape of the quantile function.

\subsection{Comparison with variance-adjusted Split Conformal and robustness of UTOPIA}
In this subsection, we present a comparison of UTOPIA with the variance-adjusted split conformal method. For simplicity of exposition, we adhere to simulation Setup 1 (Subsection \ref{sec:setup_1}). 
As mentioned there, we assume the mean function $\mu(x) = 0$ to be known and estimate the conditional variance $\var(Y \mid X = x)$ from the data. 
For the ease of the readers, we specify the details of implementation below: 
\begin{enumerate}
    \item {\bf Variance-adjusted split conformal:} Let $\hat{\sigma}^2(x)$ be estimate of $\mathrm{var}(Y \mid X = x)$. In the variance-adjusted split conformal method, the data is first split into two parts (denoted as $\mathcal{D}_1$ and $\mathcal{D}_2$). The conditional variance function is estimated using $\mathcal{D}_1$, and the score function is defined as $S(x, y) = |y|/\hat \sigma(x)$. Using $\hat{\sigma}$, the scores are calculated for $\mathcal{D}_2$, i.e., $R_i = S(X_i, Y_i)$ for all $(X_i, Y_i) \in \mathcal{D}_2$. A threshold $\hat{\tau}$ is chosen based on the $(1-\alpha)$ empirical quantile of $\{R_i : (X_i, Y_i) \in \mathcal{D}_2\}$, and the prediction interval is $\pm \hat{\tau} \hat{\sigma}(x)$.
The conditional variance function $\hat{\sigma}^2(x)$ is estimated by regressing $Y_i^2$ on $X_i$ using a random forest with a maximum depth $j$. In our experiment, we vary $j \in \{3, 4, 5, \dots, 10, 15, 20, 25, 30, 35, 40, 45, 50\}$.

\item {\bf UTOPIA:} The setup of UTOPIA remains almost identical to that specified at the beginning of the section, with the only change being in Estimator 5, which is used for estimating the conditional variance. To compare the performance with the variance-adjusted conformal predictor, here also we use random forest with maximum depth $j$ to regress $(Y - m_0(x))^2 \equiv Y^2$ on $X$, with $j \in \{3, 4, 5, \dots, 10, 15, 20, 25, 30, 35, 40, 45, 50\}$.
\end{enumerate}
Table \ref{tab:comparison} presents the results for coverage and bandwidth (median taken over 50 Monte Carlo iterations) when varying the maximum depth of the random forest regressor from $3$ to $50$. It is evident that both UTOPIA and split conformal achieve close to $95\%$ coverage, while the average width of UTOPIA is always significantly smaller. This experiment also demonstrates the robustness of UTOPIA against the misspecification or instability of some of its components. When the depth of the random forest is large, the estimator of the conditional variance function becomes complex and starts overfitting. Consequently, if the estimate of $\hat{\sigma}(x)$ is poor, the score function in the conformal prediction becomes noisy, leading to a large average width. In contrast, UTOPIA explicitly aims to minimize the width; therefore, even if some of its components result in wide intervals, it assigns smaller weights to those components, maintaining both coverage and average width. 
The last column of Table \ref{tab:comparison} represents the weight UTOPIA puts on the variance estimator ($f_5$); it is apparent from the column that as the depth of the random forest-based estimator increases, the estimator for the conditional variance becomes noisy and starts overfitting, and consequently UTOPIA puts lesser and lesser weight on this estimator. 
This simulation experiment indicates that UTOPIA is robust to the instability of some of its component prediction intervals.
\begin{table}
\centering
\caption{Comparison between UTOPIA and variance-adjusted split conformal prediction band.}
\label{tab:comparison}
\begin{tabular}{|p{1.5cm}||p{1.3cm}|p{1.3cm}|p{1.3cm}|p{1.3cm}|p{1.5cm}|  }
 \hline
 &  \multicolumn{2}{|c|}{Coverage} & \multicolumn{2}{|c|}{Bandwidth} & \\
 \hline
Max depth & UTOPIA & Conformal & UTOPIA & Conformal & Weight \\
 \hline
 3  & 0.951 & 0.950 & 5.155 & 6.233 & 0.473 \\
        4  & 0.951 & 0.952 & 5.130 & 6.009 & 0.478 \\
        5  & 0.951 & 0.951 & 5.126 & 5.880 & 0.367 \\
        6  & 0.951 & 0.951 & 5.118 & 5.853 & 0.266 \\
        7  & 0.950 & 0.952 & 5.108 & 5.937 & 0.179 \\
        8  & 0.951 & 0.951 & 5.126 & 6.167 & 0.165 \\
        9  & 0.951 & 0.950 & 5.125 & 6.471 & 0.129 \\
        10 & 0.951 & 0.950 & 5.132 & 6.871 & 0.122 \\
        15 & 0.951 & 0.951 & 5.134 & 9.463 & 0.111 \\
        20 & 0.951 & 0.951 & 5.136 & 11.403 & 0.108 \\
        25 & 0.951 & 0.951 & 5.137 & 12.641 & 0.108 \\
        30 & 0.951 & 0.950 & 5.136 & 13.153 & 0.103 \\
        35 & 0.951 & 0.950 & 5.136 & 13.297 & 0.103 \\
        40 & 0.951 & 0.950 & 5.136 & 13.322 & 0.103 \\
        45 & 0.951 & 0.950 & 5.136 & 13.325 & 0.103 \\
        50 & 0.951 & 0.950 & 5.136 & 13.325 & 0.103 \\
 \hline
\end{tabular}
\end{table}

\section{Real data analysis}
\label{sec:real_data_analysis}
In this section, we apply UTOPIA to two datasets: the Fama-French financial data \citep{fama1993common} and the FRED-MD macroeconomic datas \citep{mccracken2016fred}. We compare UTOPIA with two baseline methods: LQR and SplitCF.

\subsection{{\bf Fama-French data}}
\label{FFData}
This dataset contains annual and monthly measurements of four market risks from July 1926 until February 2023. These variables include (a) Risk-free return rate (RF), which refers to the one-month Treasury bill rate or interest rate, (b) Market factor (MKT), indicating the additional return on the market (i.e., the difference between market return and interest rate), (c) Size factor (SMB, Small Minus Big), the mean returns for small and large portfolios based on their market capitalization, and (d) Value factor (HML, High Minus Low), expressing the divergence in returns for value and growth portfolios.

We use MKT (as $x$) to predict the other three variables HML, SMB, and RF (as $y$), though it is known that they are weakly dependent. We follow the two-step estimation procedure, i.e. we first pre-train the condition mean function and the candidate estimators and then find the best linear combination of the candidate estimators (i.e. model aggregation) via solving the optimization problem. Here the conditional mean function is estimated using a two-layer neural network with 10 neurons in the hidden layer. We consider six candidate estimators which are constructed based on: 
\begin{itemize}
    \item [(1)] a two-layer neural network with width $10$ that estimates 0.05 quantile of $Y\mid X=x$; 
    \item [(2)] a three-layer neural network with width $50$ that estimates 0.35 quantile of $Y\mid X=x$; 
    \item [(3)] a quantile regression forest that estimates 0.65 quantiles of $Y\mid X=x$; 
    \item [(4)] a gradient boosting model that estimates 0.95 quantiles of $Y\mid X=x$; 
    \item [(5)] an estimate of the conditional variance $\bbE[(Y-m_0(X))^2\mid X=x]$ using random forest with number of trees 1000 and maximum depth = 5; 
    \item [(6)] a constant $1$.
\end{itemize}

We consider standardized monthly data ($n =1160$, normalized to zero mean and unit standard deviation). For each task (i.e. use $x$ = MKT to predict $y$= HML, SMB or RF), we remove the outliers with respect to $y$ using standard interquartile ranges (IQR) method (after which $n=1095, 1126, 1143$ for $y$= HML, SMB, RF, respectively). 
We randomly assigned $80\%$ of the data as the training data (which is used for pretraining estimators, solving the optimization problem, and adjusting the interval based on the user-specified coverage level ($95\%$ in our experiment)) and $20\%$ of the data as the test data (which is used for testing the performance of the obtained prediction band. Our results are shown in the following pictures and tables). We perform this experiment 100 times (each time splitting training and test data randomly) and take the average over these 100 experiments. 

\begin{figure}[H]
  \centering
\setlength{\tabcolsep}{0.2em}
  \renewcommand{\arraystretch}{0.05}
  \begin{tabular}{c c c}
    \textbf{UTOPIA} & \textbf{LQR} & \textbf{SplitCF} \\
    \begin{subfigure}{0.3\textwidth}
      \includegraphics[width=\textwidth]{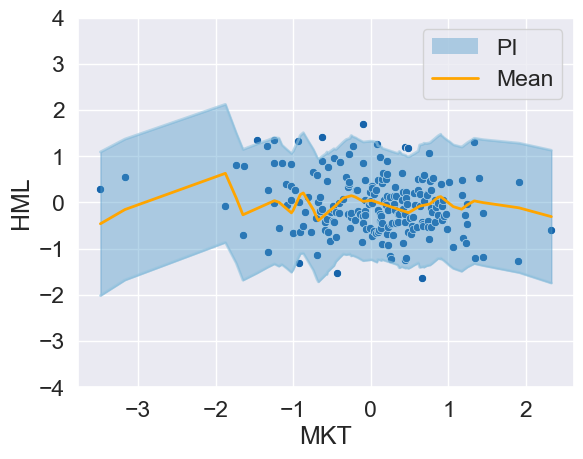}
      \label{FF:ourmethod1}
    \end{subfigure} &
    \begin{subfigure}{0.3\textwidth}
      \includegraphics[width=\textwidth]{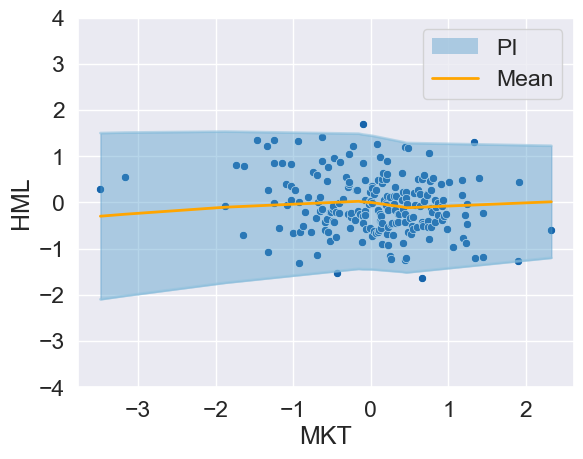}
      \label{FF:lqr1}
    \end{subfigure} &
    \begin{subfigure}{0.3\textwidth}
      \includegraphics[width=\textwidth]{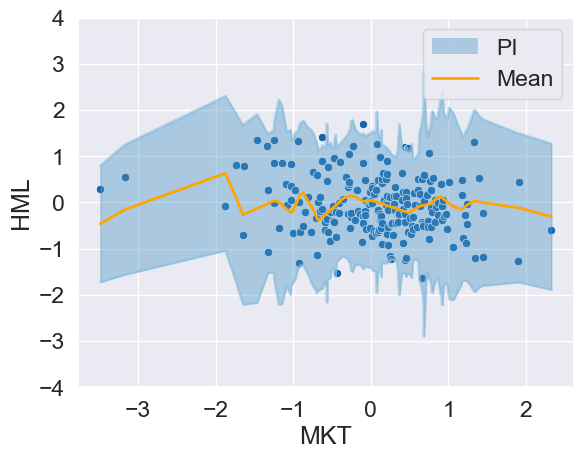}
      \label{FF:splitcf1}
    \end{subfigure} \\[-6pt]
    \begin{subfigure}{0.3\textwidth}
      \includegraphics[width=\textwidth]{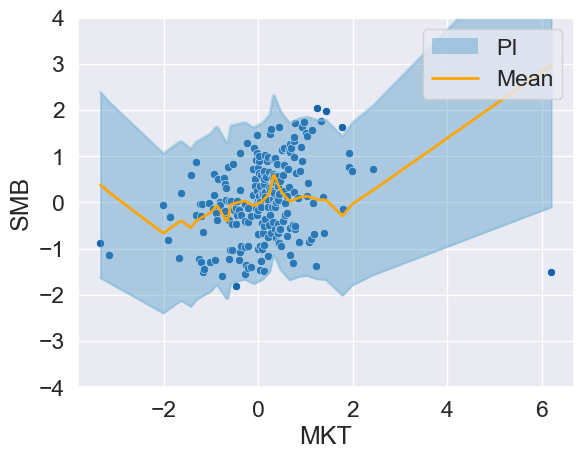}
      \label{FF:ourmethod2}
    \end{subfigure} &
    \begin{subfigure}{0.3\textwidth}
      \includegraphics[width=\textwidth]{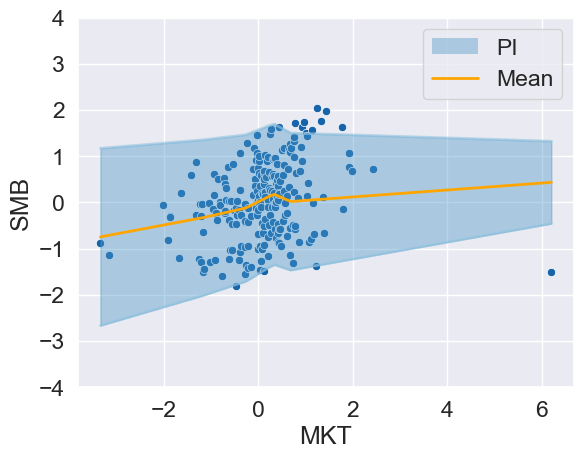}
      \label{FF:lqr2}
    \end{subfigure} &
    \begin{subfigure}{0.3\textwidth}
      \includegraphics[width=\textwidth]{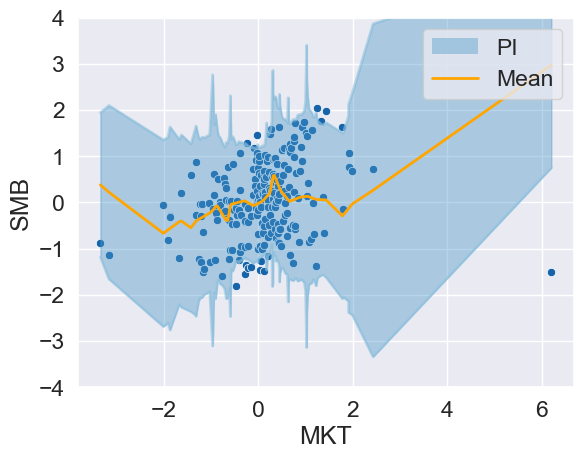}
      \label{FF:splitcf2}
    \end{subfigure} \\[-6pt]
    \begin{subfigure}{0.3\textwidth}
      \includegraphics[width=\textwidth]{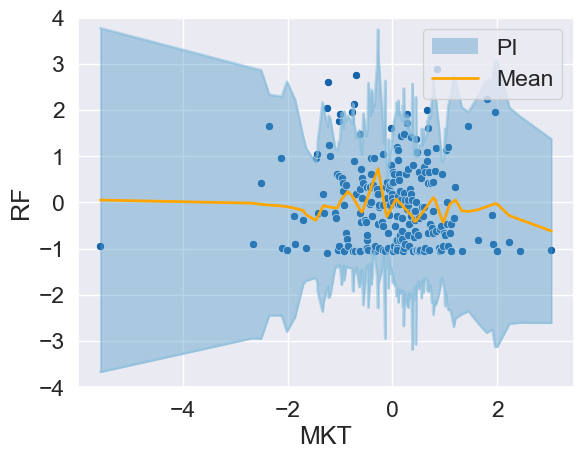}
      \label{FF:ourmethod3}
    \end{subfigure} &
    \begin{subfigure}{0.3\textwidth}
      \includegraphics[width=\textwidth]{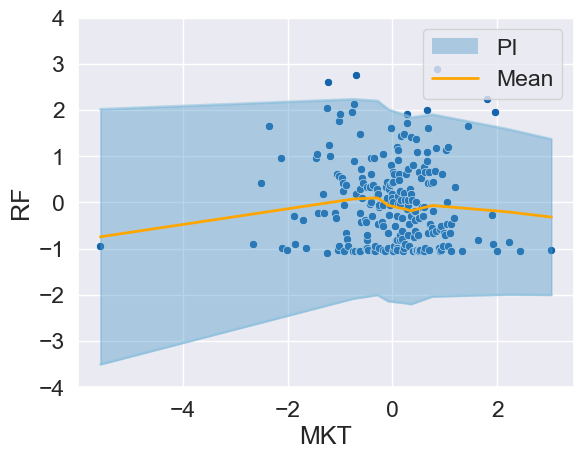}
      \label{FF:lqr3}
    \end{subfigure} &
    \begin{subfigure}{0.3\textwidth}
      \includegraphics[width=\textwidth]{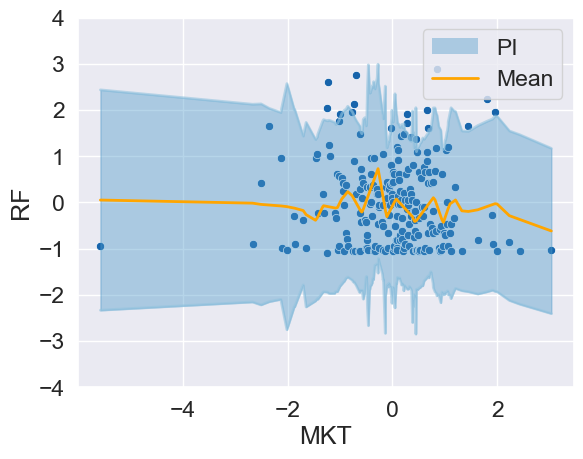}
      \label{FF:splitcf3}
    \end{subfigure}
  \end{tabular}
  \caption{The \textcolor{blue}{blue dots}, \textcolor{orange}{orange line} and \textcolor{NavyBlue}{light blue band} represent the testing data, the estimated conditional mean function and the prediction band, respectively, for the Fama-French dataset.}
  \label{FF}
\end{figure}

\vspace{-1cm}

\begin{table}[H]
\fontsize{9pt}{9pt}\selectfont
\centering
\caption{Coverage and Average Width for Fama-French Data}
\label{tab:FF}
\begin{tabular}{lcc|cc|cc}
\toprule
 & \multicolumn{2}{c}{\textbf{HML}} & \multicolumn{2}{c}{\textbf{SMB}} & \multicolumn{2}{c}{\textbf{RF}} \\
\cmidrule(lr){2-3} \cmidrule(lr){4-5} \cmidrule(lr){6-7}
 & \textbf{Coverage} & \textbf{Width} & \textbf{Coverage} & \textbf{Width} & \textbf{Coverage} & \textbf{Width} \\
\midrule
UTOPIA & $94.8\%$ & {\red $1.87$} & $94.6\%$ & {\red $2.26$} & $94.8\%$ & {\red $3.13$} \\
LQR        & $95.8\%$ & $1.98$& $95.4\%$ & $2.34$ & $97.5\%$ & $4.25$\\
SplitCF    & $95.2\%$ & $1.89$ & $95\%$ & $2.42$ & $95\%$ & $3.22$ \\
\bottomrule
\end{tabular}
\end{table}


Figure \ref{FF} and Table \ref{tab:FF} summarize the results of our experiments. All these three methods achieve close to the desired coverage $(95\%)$ for three different tasks. In comparison to the LQR and SplitCF, our approach consistently achieves a smaller bandwidth.

\subsection{{\bf FRED-MD data}}

FRED-MD (Federal Reserve Economic Data -- Macro Data) is a large-scale dataset that provides monthly economic and financial data for the United States from January 1959 to January 2023. It contains over 130 monthly macroeconomic time series, spanning a wide range of different economic sectors and sub-sectors. These variables are grouped into 8 groups: (1) output and income, (2) labor market, (3) housing, (4) consumption,
orders and inventories, (5) money and credit, (6) interest and exchange rates, (7) prices, and (8) stock
market.

In this experiment, we used one variable from the group 'Output and Income,' namely (1) RPI (Real Personal income), and another variable from the group 'Labor Market,' namely (2) USGOVT (Government employees) as our response variable $Y$. 
We aim to predict the above two variables using the rest of the variables.
For data pre-processing we mostly follow the experimental setup of \cite[Appendix C]{fan2022factor}, which, in turn, is based on \cite{mccracken2016fred}. Briefly speaking, this first transforms each variable according to the transformation code (namely TCODE) as elaborated in the Appendix of \cite{mccracken2016fred} and then removes missing values. We used the transformed data from the GitHub of \cite{fan2022factor}, namely \verb|transformed_modern.csv|. After the data pre-processing, we have $n = 453$ observations for $p = 126$ variables. 

To predict a particular variable from the rest of the variables, we use the FARM-predict technique (e.g., see \cite{zhou2022measuring}). Briefly speaking FARM-predict assumes that the covariates $X$ are generated from a factor model $X = Bf + u$ where $f$ is the underlying factor, $B$ is the loading matrix and $u$ is the idiosyncratic component. The idea is to first estimate $(\hat f, \hat B)$ and set $\hat u = X - \hat B \hat f$. We select a subset of $u$ which has a strong conditional correlation with the residual after fitting on $\hat f$, namely, $Y - \hat \bbE[Y \mid \hat f]$, where we use a linear model to estimate $\bbE[Y \mid f]$. 

In our experiment, we use two factors of $X$ and the top three components of $\hat u$ in terms of the absolute correlation with $Y - \hat \bbE[Y \mid \hat f]$. 
Upon selecting $(\hat f, \hat u)$ we use those as our predictors and employ three methods as before. We also use the same two-step estimation process and data splitting method outlined in Section \ref{FFData}. The results of our experiments are presented in Figure \ref{FRED-MD} and Table \ref{tab:FRED}. 
\begin{figure}[h]
    \centering
    \begin{subfigure}{0.47\textwidth}
        \centering
        \caption{UTOPIA for RPI}
        \includegraphics[width=0.8\textwidth]{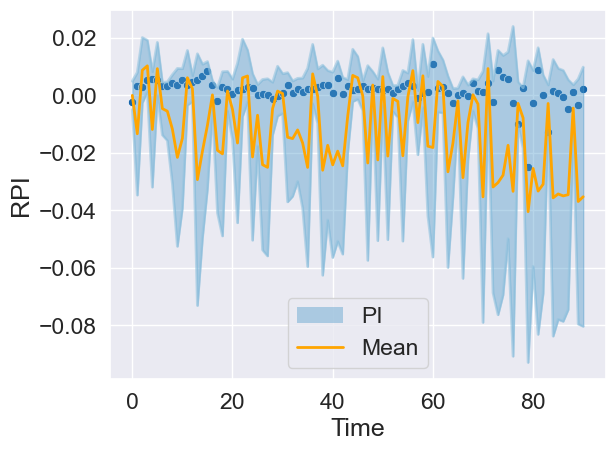}
        \label{FM:R}
    \end{subfigure}
    \hspace*{0.02\textwidth}
    \begin{subfigure}{0.47\textwidth}
        \centering
        \caption{UTOPIA for USGOVT}
        \includegraphics[width=0.8\textwidth]{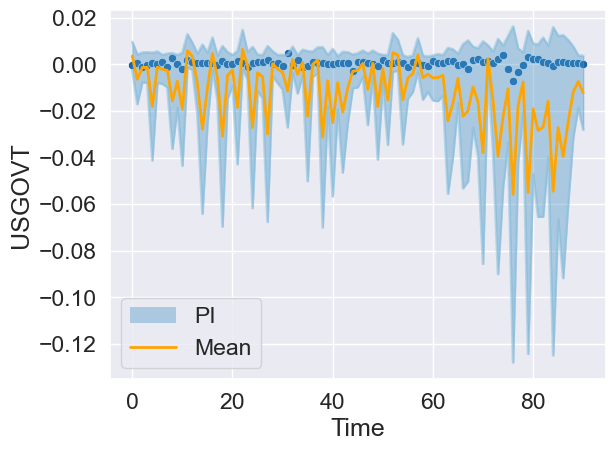}
        \label{FM:D}
    \end{subfigure}
  \caption{The \textcolor{blue}{blue dots}, \textcolor{orange}{orange line} and \textcolor{NavyBlue}{light blue band} represent the testing data, the conditional mean function and the prediction band, respectively, for the FRED-MD data.}
    \label{FRED-MD}
\end{figure}


\begin{table}[h]
\fontsize{9pt}{9pt}\selectfont
\centering
\caption{Coverage and Average Width for FRED-MD Data}
\label{tab:FRED}
\begin{tabular}{lcc|cc}
\toprule
 & \multicolumn{2}{c}{\textbf{RPI}} & \multicolumn{2}{c}{\textbf{USGOVT}} \\
\cmidrule(lr){2-3} \cmidrule(lr){4-5}
 & \textbf{Coverage} & \textbf{Width} & \textbf{Coverage} & \textbf{Width} \\
\midrule
UTOPIA & $97.8\%$ & $0.00067$ & $98.9\%$ & $0.00059$  \\
LQR        & $59.34\%$ & $0.00032$& $87.9\%$ & $0.00054$\\
SplitCF    & $97.8\%$ & $0.0011$ & $58.24\%$ & $0.00097$  \\
\bottomrule
\end{tabular}
\end{table}

As evident from Table \ref{tab:FRED}, UTOPIA attains the desired coverage ($\approx 95\%$) in both prediction tasks.  
Although the average of LQR's prediction band is slightly smaller than UTOPIA's, its coverage is much less than the desired level. On the other hand, Split conformal achieves good coverage for predicting RPI, albeit with a larger width than UTOPIA.





\section{Conclusion and future research}
In this paper, we propose a novel technique, called UTOPIA, to aggregate several candidates 
for the prediction interval of response variable $Y$ given some observable covariates $X$ with a small average width and adequate coverage guarantee. Our techniques rely on simple linear programming which is easy to train and implement. The method is universally applicable with predication interval construction from an elementary basis as one of its specific examples.  In the theoretical framework of our study, we present finite sample guarantees concerning both the average width and the coverage of our prediction interval. 
Additionally, we demonstrate the effectiveness of our methodology through experiments with both synthetic and real-world data sets.  Though our methods are testable in both statistical theory and empirical application, we emphasize the universality and trainability of our method in the acronym UTOPIA.  As illustrated in our simulation studies, the aggregation allows ensemble learning from multiple methods and distributes computation into each individual method.

The problem of minimizing average bandwidth is interesting and requires further in-depth research. We now point to a few research avenues that stem from this project. 
UTOPIA is essentially a two-shot approach, where we first capture the shape of the interval through $\hat f_\pre$ and then shrink it to attain desired coverage. An interesting future work would be to understand whether it is possible to implement this algorithm in one shot, i.e., capturing both the shape and the coverage simultaneously. Furthermore, the theoretical guarantees of our methods rely on the independence among the observations; it would also be interesting to explore how the serial dependence (e.g., the time series data) affects the coverage and width. 
Also, in one of our ongoing research, we are working on a suitable modification of UTOPIA under distribution shift, i.e., where the test/future data is generated from a distribution that is different from the historical/training data.  Another interesting extension of our approach would be to obtain a heterogeneous linear combination of the predictors (i.e. the coefficients also depend on $x$), which is expected to be more spatially adaptive than our current method.   Last but not least, it is of significant interest to allow all of the above methods to have correct conditional coverage rather than unconditional coverage.
We leave these directions for future research. 

\section{Acknowledgement}
This work is partially supported by ONR grant N00014-19-1-2120 and the NSF grants DMS-2052926,
DMS-2053832, and DMS-2210833.

\newpage
\bibliographystyle{apalike}
\bibliography{Refs}

\newpage
\appendix 
\section{Appendix: Proofs}

\subsection{Proof of Proposition \ref{prop:rc_cover}}
The proof of the corollary relies on the bound of the Rademacher complexity under various growth conditions. 
\\\\
\noindent\textbf{Polynomial growth and Exponential growth with $0 < \alpha < 2$: }As $\|f\|_\infty \le B_{\cF}$, we can take $F(x) = B_{\cF}$ to be the envelope function. Applying \cite[Lemma 2.14.1]{vaart1997weak} we obtain: 
\begin{align*}
    \cR_n(\cF)&=\bbE_{\eps, X}\left[\sup_{f \in \cF}\frac1n \sum_i \eps_i f(X_i)\right] \\
    & \le \frac{B_{\cF}}{\sqrt{n}}\sup_Q \int_0^1 \sqrt{1 + \log{\cN(\eps B_{\cF}, L_2(Q), \cF)}} \ d\eps \\
    & \le \frac{1}{\sqrt{n}}\int_0^{B_{\cF}} \sqrt{1 + \log{\cN(\eps , L_\infty, \cF)}} \ d\eps \\
    & \le \frac{1}{\sqrt{n}}\int_0^{B_{\cF}} \sqrt{1 + \alpha \log{\frac{1}{\eps}}}\ d\eps  \le C\sqrt{\frac{\alpha}{n}} \,.
\end{align*}
Now for polynomial growth we have: 
$$
\frac{1}{\sqrt{n}}\int_0^{B_{\cF}} \sqrt{1 + \log{\cN(\eps , L_\infty, \cF)}} \ d\eps \le \frac{1}{\sqrt{n}}\int_0^{B_{\cF}} \sqrt{1 + \alpha \log{\frac{1}{\eps}}}\ d\eps  \le C\sqrt{\frac{\alpha}{n}} \,.
$$
And for exponential growth with $0 \le \alpha < 2$: 
$$
\frac{1}{\sqrt{n}}\int_0^{B_{\cF}} \sqrt{1 + \log{\cN(\eps , L_\infty, \cF)}} \ d\eps \le \frac{1}{\sqrt{n}}\int_0^{B_{\cF}} \sqrt{1 + \left(\frac{1}{\eps}\right)^{\alpha}} \ d\eps \le \frac{C}{\sqrt{n}} 
$$
Here the constant $C$ depends on $({B_{\cF}}, \alpha)$. 
\\\\
\noindent\textbf{Exponential growth with $\alpha > 2$: }
For $\alpha > 2$, i.e. for non-Donsker class, the calculation is different. Denote by $\|f\|_n = \sqrt{(1/n)\sum_i f^2(X_i)}$. Fix $\tau > 0$. Let $N_\tau = \cN(\tau, L_\infty, \cF)$ and $\cF_\tau = \{f_1, \dots, f_{N_\tau}\}$ be the center of covers. Therefore for each $f$, there exists $\tau(f) \in \{1, \dots, N_\tau\}$ such that $\|f - f_{\tau(f)}\|_\infty \le \tau$. Therefore, we have: 
\begin{align}
\label{eq:tau_bound_1}
    \cR_n(\cF)&=\bbE_{\eps, X}\left[\sup_{f \in \cF}\frac1n \sum_i \eps_i f(X_i)\right] \notag \\
    & = \bbE_{\eps, X}\left[\sup_{f \in \cF}\frac1n \sum_i \eps_i \left(f(X_i) - f_{\tau(f)}(X_i) + f_{\tau(f)}(X_i)\right)\right] \notag \\
    & \le \bbE_{\eps, X}\left[\sup_{f \in \cF}\frac1n \sum_i \eps_i \left(f(X_i) - f_{\tau(f)}(X_i)\right)\right] + \bbE_{\eps, X}\left[\sup_{f \in \cF_\tau}\frac1n \sum_i \eps_i f(X_i)\right] \notag \\
    & \le \tau + \frac{1}{\sqrt{n}}\bbE_{\eps, X}\left[\max_{f \in \cF_\tau}\frac{1}{\sqrt{n}} \sum_i \eps_i f(X_i)\right] \,.
\end{align}
Now by Dudley's theorem, we have 
\begin{align*}
\bbE_{\eps, X}\left[\max_{f \in \cF_\tau}\frac{1}{\sqrt{n}} \sum_i \eps_i f(X_i)\right] & \le C\left(\int_{\frac{\tau}{4}}^{B_{\cF}} \sqrt{\log{\cN(\eps, L_\infty, \cF_{\tau})}} \ d\eps +\int^{\frac{\tau}{4}}_0 \sqrt{\log{\cN(\eps, L_\infty, \cF_{\tau})}} \ d\eps \right)\\ 
& \le C_1 \left(\int_{\frac{\tau}{4}}^{B_{\cF}} \eps^{-\frac{\alpha}{2}} \ d\eps + \int^{\frac{\tau}{4}}_0 \tau^{-\frac{\alpha}{2}} \ d\eps\right)\\
& \le C_2 \tau^{1 - \frac{\alpha}{2}} \,.
\end{align*}
Here $C_2$ depends on $({B_{\cF}},\alpha)$. Using this in \eqref{eq:tau_bound_1} we obtain: 
$$
\cR_n(\cF)\le \tau + \frac{C_2}{\sqrt{n}}\tau^{1 - \frac{\alpha}{2}}
$$
As this is true for all $\tau > 0$, we have: 
\begin{align*}
\cR_n(\cF) & \le \inf_{\tau > 0}\left\{\tau + \frac{C_2}{\sqrt{n}}\tau^{1 - \frac{\alpha}{2}}\right\} \,.
\end{align*}
Choosing $\tau = n^{-1/\alpha}$, we conclude the proof.

\subsection{Proof of Theorem \ref{thm:known_mean}}
Firstly, we prove the bandwidth guarantee \eqref{ineq:known_mean}, which mainly relies on the following proposition: 
\begin{proposition}\label{emp_process}
Let $\cF$ be a class of $B_{\cF}$-bounded functions and $\{x_i\}^n_{i=1}$ be i.i.d samples drawn from some distribution. Then for any $t>0$, with probability at least $1-2e^{-t}$, we have
\begin{align*}
    \left|\sup_{f\in\cF}\left(\frac1n\sum_if(x_i)-\bbE[f(X)]\right)\right|\leq 2\cR_n(\cF)+2B_{\cF}\sqrt{\frac{t}{2n}}
\end{align*}
\end{proposition}
\begin{proof}[Proof of Proposition \ref{emp_process}]
By McDiarmid's inequality, we know with probability at least $1-e^{-t}$,
\begin{align*}
 \sup_{f\in\cF}\left(\frac1n\sum_if(x_i)-\bbE[f(X)]\right)\leq \bbE_{X} \left[\sup_{f\in\cF}\left(\frac1n\sum_if(x_i)-\bbE[f(X)]\right)\right] + 2B_{\cF}\sqrt{\frac{t}{2n}}.
\end{align*}
By Gin{\'e}-Zinn symmetrization, we further have
\begin{align*}
\bbE_{X} \left[\sup_{f\in\cF}\left(\frac1n\sum_if(x_i)-\bbE[f(X)]\right)\right]\leq 2 \bbE_{\eps, X}\left[\sup_{f \in \cF}\frac1n \sum_i \eps_i f(X_i)\right]=2\cR_n(\cF).
\end{align*}
Consequently, we have with probability at least $1-e^{-t}$,
\begin{align*}
    \sup_{f\in\cF}\left(\frac1n\sum_if(x_i)-\bbE[f(X)]\right)\leq 2\cR_n(\cF)+2B_{\cF}\sqrt{\frac{t}{2n}}.
\end{align*}
Similarly, we have with probability at least $1-e^{-t}$,
\begin{align*}
    \sup_{f\in\cF}\left(\bbE[f(X)]-\frac1n\sum_if(x_i)\right)\leq 2\cR_n(\cF)+2B_{\cF}\sqrt{\frac{t}{2n}}.
\end{align*}
Thus, we prove Proposition \ref{emp_process}.
\end{proof}
With Propostion \ref{emp_process} in hand, we now prove \eqref{ineq:known_mean} in the following. Suppose $f^\star$ is minimizer of \eqref{population_opt:known_mean} over $\cF$, i.e. 
$$
f^\star = \argmin_{f \in \cF: f \geq f_0}\bbE[f(X)]  \,.
$$
Then we have $\Delta(\cF) = \bbE[f^\star(X)] - \bbE[f_0(X)]$. Moreover, $f^{\star}$ is a feasible solution of \eqref{sample_opt:known_mean}. By the optimality of $\hat f$, it holds that
\begin{align*}
 \frac1n\sum_i\hat f_\pre (x_i)\leq\frac1n\sum_i f^{\star}(x_i),   
\end{align*}
which implies
\begin{align*}
    &\bbE[\hat f_\pre(X)\mid \cD]-\bbE[f^{\star}(X)]\\
    &\quad = \bbE[\hat f_\pre(X)\mid \cD]-\frac1n\sum_i\hat f_\pre (x_i)+\frac1n\sum_i\hat f_\pre (x_i)-\frac1n\sum_i f^{\star}(x_i)+\frac1n\sum_i f^{\star}(x_i)-\bbE[f^{\star}(X)]\\
    &\quad\leq \sup_{f\in\cF}\left(\bbE[f(X)]-\frac1n\sum_if(x_i)\right)+\sup_{f\in\cF}\left(\frac1n\sum_if(x_i)-\bbE[f(X)]\right)\\
    &\quad\leq 4\cR_n(\cF)+4B_{\cF}\sqrt{\frac{t}{2n}}.
\end{align*}
Here the last inequality follows from Proposition \ref{emp_process}. Consequently, we have
\begin{align*}
\bbE[\hat f_\pre (X)\mid \cD] -\bbE[f_0(X)] &=\Delta(\cF)+\bbE[\hat f_\pre (X)\mid \cD]-\bbE[f^{\star}(X)]\\
&\leq \Delta(\cF)+4\cR_n(\cF)+4B_{\cF}\sqrt{\frac{t}{2n}}.
\end{align*}
Thus, we prove \eqref{ineq:known_mean}. The coverage guarantee \eqref{cov:vc_known_mean} and \eqref{cov:known_mean} can be seen as a special case of Theorem \ref{thm:one_step} with $\cM=\{m_0\}$. Thus we do not go into details here.

\subsection{Proofs of Theorem \ref{thm:one_step}}
\begin{proof}

Recall that given any guess $m$ of $m_0$, we denote by $f_m(x)$ the supremum of $(Y-m(X))^2\mid X=x$, i.e. 
$$
f_m(x):=\inf_{t}\left\{t\,\big|\, \bbP\left((Y-m(X))^2\leq t\big|X=x\right)=1\right\}
$$ 
and we have $f_0(x)=f_{m_0}(x)$.
And we denote 
\begin{align*}
(m^{\star},f^{\star})=\argmin_{m\in\cM,f\in\cF:f\geq f_m}\bbE[f(X)].
\end{align*}
Then we have $\Delta(\cM,\cF) = \bbE[f^\star(X)] - \bbE[f_0(X)]$. Moreover, $(m^{\star},f^{\star})$ is a feasible solution of \eqref{sample_opt:general_form}. By the optimality of $(\hat m, \hat f_{\pre})$, we know that
\begin{align*}
\frac1n\sum^n_{i=1}\hat f_{\pre}(x_i)\leq \frac1n\sum^n_{i=1}f^{\star}(x_i).
\end{align*}
It then holds that
\begin{align*}
&\bbE[\hat f_\pre (X)\mid\cD]-\bbE[f^{\star}(X)]\\
&\quad=\bbE[\hat f_\pre (X)\mid\cD]-\frac1n\sum^n_{i=1}\hat f_\pre(x_i)+\frac1n\sum^n_{i=1}\hat f_\pre(x_i)-\frac1n\sum^n_{i=1}f^{\star}(x_i)+\frac1n\sum^n_{i=1}f^{\star}(x_i)-\bbE[f^{\star}(X)]\\
&\quad\leq \sup_{f\in\cF}\left(\bbE[f(X)]-\frac1n\sum_if(x_i)\right)+\sup_{f\in\cF}\left(\frac1n\sum_if(x_i)-\bbE[f(X)]\right)\\
&\quad\leq 4\cR_n(\cF)+4B_{\cF}\sqrt{\frac{t}{2n}}.
\end{align*}
Here the last inequality follows from Propostion \ref{emp_process}. Consequently, we have
\begin{align*}
\bbE[\hat f_\pre(X)\mid \cD] -\bbE[f_0(X)] &=\Delta(\cM,\cF)+\bbE[\hat f(X)\mid \cD]-\bbE[f^{\star}(X)]\\
&\leq \Delta(\cM,\cF)+4\cR_n(\cF)+4B_{\cF}\sqrt{\frac{t}{2n}}.
\end{align*}
We then prove the coverage guarantees \eqref{ineq:vc_coverage} and \eqref{ineq:coverage} in the following.
\\\\
\noindent\textbf{Case 1: $\cG$ is VC class.}
We denote $\hat g(x,y)=\hat f_\pre(x)-(y-\hat m(x))^2$. It then holds that
\begin{align*}
\bbP\left((Y - \hat m(X))^2 \le \hat f_\pre(X)\,\big|\,\cD\right)&=\bbE\left[\mathbf{1}_{\{\hat g(X,Y)\geq 0\}}\,\big|\,\cD\right]\\
&\geq \frac1n\sum^n_{i=1}\mathbf{1}_{\{\hat g(x_i,y_i)\geq 0\}}-\sup_{g\in\cG}\left|\frac1n\sum^n_{i=1}\mathbf{1}_{\{g(x_i,y_i)\geq 0\}}-\bbE\left[\mathbf{1}_{\{g(X,Y)\geq 0\}}\right]\right|\\
&= 1-\sup_{g\in\cG}\left|\frac1n\sum^n_{i=1}\mathbf{1}_{\{g(x_i,y_i)\geq 0\}}-\bbE\left[\mathbf{1}_{\{g(X,Y)\geq 0\}}\right]\right|.
\end{align*}
Since $\cG$ is VC class, we know $\{\cG\geq 0\}=\{\{g\geq 0\}: g\in\cG\}$ is also VC class with VC dimension $\leq \VC(\cG)$ (see Lemma 2.6.18 in \citep{van1996weak}). By symmetrization, Haussler's bound and Proposition \ref{prop:rc_cover} with polynomial growth, we have
\begin{align*}
\bbE\left[\sup_{g\in\cG}\left|\frac1n\sum^n_{i=1}\mathbf{1}_{\{g(x_i,y_i)\geq 0\}}-\bbE\left[\mathbf{1}_{\{g(X,Y)\geq 0\}}\right]\right|\right]\leq   c\sqrt{\frac{\VC(\cG)}{n}}
\end{align*}
for some constant $c$. Then Theorem 2.3 in \cite{bousquet2002bennett} yields with probability at least $1-e^{-t}$:
\begin{align*}
\sup_{g\in\cG}\left|\frac1n\sum^n_{i=1}\mathbf{1}_{\{g(x_i,y_i)\geq 0\}}-\bbE\left[\mathbf{1}_{\{g(X,Y)\geq 0\}}\right]\right|\leq c\sqrt{\frac{\VC(\cG)}{n}}+c\sqrt{\frac{t}{n}}.
\end{align*}
Therefore we show that with probability at least $1-e^{-t}$,
\begin{align*}
\bbP\left((Y - \hat m(X))^2 \le \hat f_\pre(X)\,\big|\,\cD\right) \geq 1-c\sqrt{\frac{\VC(\cG)}{n}}-c\sqrt{\frac{t}{n}}
\end{align*}
for some constants $c$.

\noindent\textbf{Case 2: $\cG$ is non-VC class.}
To prove the coverage guarantee for non-VC class $\cG$, we define the following function
\begin{align*}
h_{\delta}(t):=
    \begin{cases}
        1 & \text{if } t \leq -\delta \\
        -\frac{t}{\delta} & \text{if } -\delta <t<0\\
        0 &\text{if } t\geq 0.
    \end{cases}
\end{align*}
Note that 
\begin{align*}
\mathbf{1}_{\{t< -\delta\}}\leq h_{\delta}(t)\leq \mathbf{1}_{\{t< 0\}}, \ \forall t\in\bbR.
\end{align*}
As before, we denote $\hat g(x,y)=\hat f_\pre(x)-(y-\hat m(x))^2$. It then holds that
\begin{align*}
\bbP\left((Y - \hat m(X))^2 > \hat f_\pre(X)+\delta\,\big|\,\cD\right) &=\bbE\left[\mathbf{1}_{\{\hat g(X,Y)< -\delta\}}\,\big|\,\cD\right]\\
&\leq \bbE\left[h_{\delta}(\hat g(X,Y))\,\big|\,\cD\right]\\
&\leq\frac1n\sum^n_{i=1}h_{\delta}(\hat g(x_i,y_i))+\sup_{g\in\cG}\left( \bbE\left[h_{\delta}(g(X,Y))\right]-\frac1n\sum^n_{i=1}h_{\delta}\left( g(x_i,y_i)\right)\right).
\end{align*}
Since $(\hat m,\hat f_{\pre})$ is feasible for \eqref{sample_opt:general_form}, we have
\begin{align*}
\frac1n\sum^n_{i=1}h_{\delta}(\hat g(x_i,y_i))\leq \frac1n\sum^n_{i=1}\mathbf{1}_{\{\hat g(x_i,y_i)< 0\}}=0.
\end{align*}
By the definition of $h_{\delta}$, we know $\|h_{\delta}\|_{\infty}\leq 1$ and $h_{\delta}$ is $\delta^{-1}$-Lipschitz. Thus, we have with probability at least $1-e^{-t}$,
\begin{align*}
&\sup_{g\in\cG}\left( \bbE\left[h_{\delta}(g(X,Y))\right]-\frac1n\sum^n_{i=1}h_{\delta}\left( g(x_i,y_i)\right)\right)\\
 &\quad\leq 2\bbE_{\cD}\bbE_{\mathbf{\eps}}\left[\sup_{g\in\cG}\frac1n\sum^n_{i=1}\eps_i h_{\delta}\circ g(x_i,y_i)\right|]+\sqrt{\frac{t}{2n}} \\
     &\quad\leq \frac{2}{\delta}\cdot\bbE_{\cD}\bbE_{\mathbf{\eps}}\left[\sup_{g\in\cG}\frac1n\sum^n_{i=1}\eps_i g(x_i,y_i)\right]+\sqrt{\frac{t}{2n}}\\
     &\quad = \frac{2}{\delta}\cdot\cR_n(\cG)+\sqrt{\frac{t}{2n}},
\end{align*}
where the first inequality follows from the McDiarmid's inequality and the symmetrization argument and the second inequality follows from Talagrand’s contraction
lemma \citep{ledoux1991probability}. Consequently, we have with probability at least $1-e^{-t}$,
\begin{align*}
 \bbP\left((Y - \hat m(X))^2 > \hat f_\pre(X)+\delta\,\big|\,\cD\right)\leq   \frac{2}{\delta}\cdot\cR_n(\cG)+\sqrt{\frac{t}{2n}}.
\end{align*}
This completes the proof. 
\end{proof}

\subsection{Proofs of Theorem \ref{thm:two_step}}
\begin{proof}
Suppose that $Y\in\cB(0,R)$ and $\cM$ is $B_{\cM}$-bounded. Note that 
\begin{align*}
    \left|\left(y_i-\hat m(x_i)\right)^2-\left(y_i-m_0(x_i)\right)^2\right|&=\left|\left(2y_i-\hat m(x_i)-m_0(x_i)\right)\left(\hat m(x_i)-m_0(x_i)\right)\right|\\
    &\leq (B_{\cM}+3R)\cdot\|\hat m-m_0\|_{\infty}.
\end{align*}
Thus, we have
\begin{align*}
\frac1n\sum^n_{i=1}\hat f_\pre(x_i)&=\inf_{\substack{f\in\cF\\ f(x_i)\geq (y_i-\hat m(x_i))^2}}\frac1n\sum^n_{i=1}f(x_i)\\
&\leq \inf_{\substack{f\in\cF\\ f(x_i)\geq (y_i-m_0(x_i))^2+(B_{\cM}+3R)\cdot\|\hat m-m_0\|_{\infty}}}\frac1n\sum^n_{i=1}f(x_i)
\end{align*}
For $f^{\star}=\argmin_{f\in\cF:f\geq f_0}\bbE[f(X)]$, we have 
$f^{\star}(x_i)\geq f_0(x_i)\geq (y_i-m_0(x_i))^2$ for any $i=1\ldots n$. If $(B_{\cM}+3R)\cdot\|\hat m-m_0\|_{\infty}\leq D$, then by Assumption \ref{assm}, we have $f^{\star}+(B_{\cM}+3R)\cdot\|\hat m-m_0\|_{\infty}\in\cF$. Furthermore, we have
$$
f^{\star}(x_i)+(B_{\cM}+3R)\cdot\|\hat m-m_0\|_{\infty}\geq (y_i-m_0(x_i))^2+(B_{\cM}+3R)\cdot\|\hat m-m_0\|_{\infty}.
$$
Consequently, we have
\begin{align*}
\frac1n\sum^n_{i=1}\hat f_\pre(x_i)
&\leq \inf_{\substack{f\in\cF\\ f(x_i)\geq (y_i-m_0(x_i))^2+(B_{\cM}+3R)\cdot\|\hat m-m_0\|_{\infty}}}\frac1n\sum^n_{i=1}f(x_i)\\
&\leq \frac1n\sum^n_{i=1}f^{\star}(x_i)+(B_{\cM}+3R)\cdot\|\hat m-m_0\|_{\infty}.
\end{align*}
Thus, it holds with probability at least $1-2e^{-t}$ that
\begin{align*}
\bbE[\hat f_\pre(X) \mid \cD] -\bbE[f^{\star}(X)]&\leq \frac1n\sum^n_{i=1}\hat f_\pre(x_i)-\frac1n\sum^n_{i=1}f^{\star}(x_i)+2\left|\sup_{f\in\cF}\left(\frac1n\sum_if(x_i)-\bbE[f(X)]\right)\right|\\
&\leq (B_{\cM}+3R)\cdot\|\hat m-m_0\|_{\infty}+4\cR_n(\cF)+4B_{\cF}\sqrt{\frac{t}{2n}}.
\end{align*}
By the definition of $f^{\star}$, we have
\begin{align*}
\bbE[\hat f_\pre(X) \mid \cD] -\bbE[f_0(X)]&=\Delta(\cF)+\bbE[\hat f_\pre(X) \mid \cD] -\bbE[f^{\star}(X)]\\
&\leq\Delta(\cF) +(B_{\cM}+3R)\cdot\|\hat m-m_0\|_{\infty}+4\cR_n(\cF)+4B_{\cF}\sqrt{\frac{t}{2n}}.
\end{align*}
The coverage guarantee can be seen as a special case of Theorem \ref{thm:one_step} with $\cM=\{\hat m\}$. Thus we do not go into details here.

\end{proof}

\subsection{Proof of Theorem \ref{thm:adjust}}
Recall that
\begin{align*}
p_{n_1} := \bbP\left((Y - \hat m(X))^2 \le \hat f_\pre(X) + \delta \,\big|\,\cD_1\right)
\end{align*}
and
\begin{align*}
    \cI := \left\{i \in \{1, \dots, n_2\}: (y_i - \hat m(x_i))^2 \le \hat f_\pre(x_i) + \delta\right\}.
\end{align*}
The following lemma establishes a lower bound on $|\cI|$: 
\begin{lemma}
\label{card_I}
For any $\beta>0$, with probability at least $1-e^{-\frac{n_2 p_{n_1}}{2n^{2\beta}_1}}$, we have $|\, \cI|\geq (1-{n_1}^{-\beta})p_{n_1} {n_2}$.
\end{lemma}
\begin{proof}
We define a sequence of ${n_2}$ random variables $\{Z_1, \dots, Z_{n_2}\}$ as follows
$$
Z_i = \mathds{1}((X_i, Y_i) \in \cI)  = \mathds{1}((Y_i - \hat m(X_i))^2 \le \hat f_\pre(X_i) + \delta).
$$
Conditional on $\cD_1$, $Z_i \overset{i.i.d.}{\sim} \mathrm{Ber}(p_{n_1})$ and consequently $|\, \cI|=\sum^{n_2}_{i=1} Z_i \mid \cD_1 \sim \mathrm{Bin}({n_2}, p_{n_1})$. By the Chernoff bound for binomial distribution, we have
$$
|\, \cI|\geq (1-{n_1}^{-\beta})p_{n_1} {n_2},  \ \ \text{w.p. greater than } 1-e^{-\frac{p_{n_1} {n_2}}{2n^{2\beta}_1}} \,.
$$
\end{proof}

With Lemma \ref{card_I} in hand, we are ready to prove Theorem \ref{thm:adjust}.
\begin{proof}[Proof of Theorem \ref{thm:adjust}]
 For notation simplicity, we define the function $g_\lambda(x, y)$ as: 
\begin{align*}
g_{\lambda}(x,y):=\mathbf{1}_{\{(y-\hat m(x))^2>\lambda(\hat f_\pre(x)+\delta) \}}.
\end{align*}  
Then we have 
\begin{align*}
& \bbP\left((Y - \hat m(X))^2 > \hat \lambda(\alpha)(\hat f_\pre(X) + \delta) \,\big|\,\cD_1, \cD_2\right)\\
&\quad=\bbE[g_{\hat\lambda(\alpha)}(X,Y)\mid \cD_1, \cD_2]\\
&\quad=\frac{1}{|\cI|}\sum_{i\in\cI}g_{\hat\lambda(\alpha)}(x_i,y_i)+\bbE[g_{\hat\lambda(\alpha)}(X,Y)\mid \cD_1, \cD_2]-\frac{1}{|\cI|}\sum_{i\in\cI}g_{\hat\lambda(\alpha)}(x_i,y_i).
\end{align*}
Note that by the choice of ${\hat\lambda(\alpha)}$, we have
\begin{align*}
    \alpha-\frac{1}{|\cI|}\leq\frac{1}{|\cI|}\sum_{i\in\cI}g_{\hat\lambda(\alpha)}(x_i,y_i)\leq \alpha.
\end{align*}
Consequently, we have
\begin{align*}
    &\left|\bbP\left((Y - \hat m(X))^2 > \hat \lambda(\alpha)(\hat f_\pre(X) + \delta) \,\big|\,\cD_1, \cD_2\right)-\alpha\right|\\
   &\quad \leq \frac{1}{|\cI|}+\left|\frac{1}{|\cI|}\sum_{i\in\cI}g_{\hat\lambda(\alpha)}(x_i,y_i)-\bbE[g_{\hat\lambda(\alpha)}(X,Y)\mid \cD_1, \cD_2]\right|.
\end{align*}
By Lemma \ref{card_I}, we know the first term is roughly $1/{n_2}$ when ${n_1}$ is large enough. In the sequel, we upper bound the second term.
Note that
\begin{align*}
&\left|\frac{1}{|\cI|}\sum_{i\in\cI}g_{\hat\lambda(\alpha)}(x_i,y_i)-\bbE[g_{\hat\lambda(\alpha)}(X,Y)\mid \cD_1, \cD_2]\right|\\
&\leq\underbrace{\left|\frac{{n_2}}{|\cI|}\cdot\frac{1}{n_2}\sum_{i\in\cI}g_{\hat\lambda(\alpha)}(x_i,y_i)-\frac{1}{n_2}\sum_{i\in\cI}g_{\hat\lambda(\alpha)}(x_i,y_i)\right|}_{\rm (i)}+\underbrace{\left|\frac{1}{n_2}\sum_{i\in\cI}g_{\hat\lambda(\alpha)}(x_i,y_i)-\frac{1}{n_2}\sum^{n_2}_{i=1}g_{\hat\lambda(\alpha)}(x_i,y_i)\right|}_{\rm (ii)}\\
&\quad+\underbrace{\left|\frac{1}{n_2}\sum^{n_2}_{i=1}g_{\hat\lambda(\alpha)}(x_i,y_i)-\bbE[g_{\hat\lambda(\alpha)}(X,Y)\mid \cD_1, \cD_2]\right|}_{\rm (iii)}.
\end{align*}
\noindent\textbf{Upper bound of (i):}
For (i), we have
\begin{align*}
{\rm (i)} =\left|\frac{{n_2}}{|\cI|}-1\right|\left|\frac{1}{n_2}\sum_{i\in\cI}g_{\hat\lambda(\alpha)}(x_i,y_i)\right|\leq \left|\frac{{n_2}}{|\cI|}-1\right|\cdot\frac{|\cI|}{{n_2}}=1-\frac{|\cI|}{{n_2}},
\end{align*}
where the inequality follows from the fact that $\|g_{\hat\lambda(\alpha)}\|_{\infty}\leq 1$.

\noindent\textbf{Upper bound of (ii):}
For (ii), we have
\begin{align*}
{\rm (ii)}=\frac{1}{n_2}\sum_{i\notin\cI}g_{\hat\lambda(\alpha)}(x_i,y_i)\leq \frac{1}{n_2}\cdot({n_2}-|\cI|)=1-\frac{|\cI|}{{n_2}}.
\end{align*}

\noindent\textbf{Upper bound of (iii):}
For (iii), we have w.p. $\geq 1-e^{-t}$
\begin{align*}
{\rm (iii)}&\leq \sup_{\lambda\geq 0}\left|\frac{1}{n_2}\sum^{n_2}_{i=1}g_{\lambda}(x_i,y_i)-\bbE[g_{\lambda}(X,Y)\mid \cD_1, \cD_2]\right|\\
&= \sup_{\lambda\geq 0}\left|\frac{1}{n_2}\sum^{n_2}_{i=1}g_{\lambda}(x_i,y_i)-\bbE[g_{\lambda}(X,Y)\mid \cD_1]\right|\\
&\leq 2\bbE_{\eps,\cD_2}\left[\sup_{\lambda\geq 0}\frac{1}{n_2}\sum^{n_2}_{i=1}\eps_i \cdot g_{\lambda}(x_i,y_i)\,\bigg|\,\cD_1\right]+\sqrt{\frac{t}{2n_2}},
\end{align*}
where the last inequality follows from a simple application of McDiarmid’s inequality and Gin{\'e}-Zinn symmetrization. Applying \cite[Lemma 2.14.1]{vaart1997weak} we obtain: 
\begin{align*}
    \bbE_{\eps,\cD_2}\left[\sup_{\lambda\geq 0}\frac{1}{n_2}\sum^{n_2}_{i=1}\eps_i \cdot g_{\lambda}(x_i,y_i)\,\bigg|\,\cD_1\right]\lesssim \frac{1}{\sqrt{{n_2}}}\sup_Q \int_0^1 \sqrt{\log{N(\eps, L_2(Q), \cG_{\lambda})}} \ d\eps,
\end{align*}
where $\cG_{\lambda}:=\{g_{\lambda}\mid \lambda\geq 0\}$. Note that we define the function class $\cG_\lambda$ conditioning on $\cD_1$, i.e. fixing $(\hat m, \hat f_\pre)$. By \cite[Theorem 2.6.7]{vaart1997weak}, we have for any $0<\eps<1$ that
\begin{align*}
\sup_{Q}N(\eps, L_2(Q), \cG_{\lambda})\leq c V(\cG_{\lambda})(16e)^{V(\cG_{\lambda})}\left(\frac{1}{\eps}\right)^{V(\cG_{\lambda})-1},
\end{align*}
where $c$ is an universal constant and $V(\cG_{\lambda})$ is the VC dimension of the function class $\cG_{\lambda}$. Notice that
$$
\cG_{\lambda}=\{\mathbf{1}_{\{\lambda(\hat f_\pre(x)+\delta)-(y-\hat m(x))^2<0\}}\,|\,\lambda\geq 0\},
$$
which implies $V(\cG_{\lambda})\leq 2$. Consequently, we have
\begin{align*}
\bbE_{\eps,\cD_2}\left[\sup_{\lambda\geq 0}\frac{1}{n_2}\sum^{n_2}_{i=1}\eps_i \cdot g_{\lambda}(x_i,y_i)\,\bigg|\,\cD_1\right]&\lesssim  \frac{1}{\sqrt{{n_2}}} \sup_Q \int_0^1 \sqrt{\log{N(\eps, L_2(Q), \cG_{\lambda})}} \ d\eps\\
   &\leq \frac{1}{\sqrt{{n_2}}} \int_0^1 \sqrt{\log\left\{c V(\cG_{\lambda})(16e)^{V(\cG_{\lambda})}\left(\frac{1}{\eps}\right)^{V(\cG_{\lambda})-1}\right\}} \ d\eps\\
   &\leq \frac{1}{\sqrt{{n_2}}} \int_0^1 \sqrt{ \log\left\{2c(16e)^{2}\left(\frac{1}{\eps}\right)\right\}} \ d\eps \lesssim \frac{1}{\sqrt{{n_2}}}\,.
\end{align*}
As a result, we have w.p. $\geq 1-e^{-t}$
\begin{align*}
{\rm (iii)}\leq 2\bbE_{\eps,\cD_2}\left[\sup_{\lambda\geq 0}\frac{1}{n_2}\sum^{n_2}_{i=1}\eps_i \cdot g_{\lambda}(x_i,y_i)\,\bigg|\,\cD_1\right]+\sqrt{\frac{t}{2n_2}}\leq c\cdot\sqrt{\frac{t}{{n_2}}}
\end{align*}
for some constants $c$.

Combining the upper bounds for (i), (ii) and (iii), we have w.p. $\geq 1-e^{-t}$
\begin{align*}
\left|\frac{1}{|\cI|}\sum_{i\in\cI}g_{\hat\lambda(\alpha)}(x_i,y_i)-\bbE[g_{\hat\lambda(\alpha)}(X,Y)\mid \cD_1, \cD_2]\right|\leq 2\left( 1-\frac{|\cI|}{{n_2}}\right)+c\cdot\sqrt{\frac{t}{{n_2}}},
\end{align*}
which implies
\begin{align*}
&\left|\bbP\left((Y - \hat m(X))^2 > \hat \lambda(\alpha)(\hat f_\pre(X) + \delta) \,\big|\,\cD_1, \cD_2\right)-\alpha\right|\\
   &\quad \leq \frac{1}{|\cI|}+\left|\frac{1}{|\cI|}\sum_{i\in\cI}g_{\hat\lambda(\alpha)}(x_i,y_i)-\bbE[g_{\hat\lambda(\alpha)}(X,Y)\mid \cD_1, \cD_2]\right|\\
   &\quad\leq  \frac{1}{|\cI|}+2\left( 1-\frac{|\cI|}{{n_2}}\right)+c\cdot\sqrt{\frac{t}{{n_2}}}.
\end{align*}
By Lemma \ref{card_I}, with probability at least $1-e^{-\frac{p_{n_1} {n_2}}{2n^{2\beta}_1}}$, we have $|\, \cI|\geq (1-{n_1}^{-\beta})p_{n_1} {n_2}$ for any $\beta>0$. Consequently, for any $t, \beta>0$, with probability at least $1-e^{-\frac{p_{n_1} {n_2}}{2n^{2\beta}_1}}-e^{-t}$, we have
\begin{align*}
&\left|\bbP\left((Y - \hat m(X))^2 > \hat \lambda(\alpha)(\hat f_\pre(X) + \delta) \,\big|\,\cD_1, \cD_2\right)-\alpha\right|\\
&\quad\leq \frac{1}{(1-{n_1}^{-\beta})p_{n_1} {n_2}}+2(1-p_{n_1}+{n_1}^{-\beta}p_{n_1})+c\cdot\sqrt{\frac{t}{{n_2}}}.
\end{align*}
We further choose $(\beta,t)$ such that $p_{n_1}{n_2}/2n^{2\beta}_1=t$ and $t=10\cdot\log {n_2}$. Then with probability at least $1-2{n_2}^{-10}$, we have
\begin{align*}
&\left|\bbP\left((Y - \hat m(X))^2 > \hat \lambda(\alpha)(\hat f_\pre(X) + \delta) \,\big|\,\cD_1, \cD_2\right)-\alpha\right|\\
&\quad\leq \frac{1}{\left(1-\sqrt{\frac{20\log {n_2}}{p_{n_1}{n_2}}}\right)p_{n_1} {n_2}}+2(1-p_{n_1})+c_1\cdot\sqrt{\frac{\log {n_2}}{{n_2}}}.
\end{align*}
When ${n_1},{n_2}$ are large enough such that
\begin{align*}
p_{n_1}\geq \frac12,\quad\sqrt{\frac{40\log {n_2}}{{n_2}}}\leq \frac12,
\end{align*}
we have
\begin{align*}
\left|\bbP\left((Y - \hat m(X))^2 > \hat \lambda(\alpha)(\hat f_\pre(X) + \delta) \,\big|\,\cD_1, \cD_2\right)-\alpha\right|&\leq \frac{4}{{n_2}}+2(1-p_{n_1})+c_1\cdot\sqrt{\frac{\log {n_2}}{{n_2}}}\\
&\leq 2(1-p_{n_1})+c_2\cdot\sqrt{\frac{\log {n_2}}{{n_2}}}.
\end{align*}
Thus, we finish the proof.
\end{proof}

\subsection{Proof of Corollary \ref{rc:aug}}
As $\cF$ is the linear span of $K$ many functions, we have $\VC(\cF) \le K+2$. Then by Haussler's bound and Proposition \ref{prop:rc_cover} with polynomial growth, we have
\begin{align*}
\cR_n(\cF)  = \bbE\left[\sup_{f \in \cF} \frac1n \sum_i \eps_i f(X_i)\right] \lesssim \sqrt{\frac{\VC(\cF) }{n}} \lesssim \sqrt{\frac{K}{n}} \,.
\end{align*}
As 
$$
\cG=\left\{\sum^K_{j=1}\alpha_j f_j(x)-\left(y-\sum^L_{j=1}\beta_j m_j(x)\right)^2\,\bigg|\, \alpha_j\geq 0 ,\beta_j \in\mathbb{R}\right\}
$$ 
is the linear span of $\cO(K+L^2)$ many functions, we have $\VC(\cG) \lesssim K+L^2$. Thus we have 
\begin{align*}
\cR_n(\cG) & = \bbE\left[\sup_{g \in \cG} \frac1n \sum_i \eps_i g(X_i)\right] \lesssim\sqrt{\frac{\VC(\cG)}{n}}\lesssim \sqrt{\frac{K+L^2}{n}} \,.    
\end{align*}
Then Coroallry \ref{rc:aug} directly follows from Theorem \ref{thm:one_step}.

\subsection{Proof of Corollary \ref{cor:smooth_func}}
By Lemma 6.6 of \cite[2022 version]{sen2018gentle} we have: 
$$
\log{N(\eps, C_L^k([0, 1]^d), \| \cdot \|_\infty)} \le c\cdot \eps^{-\frac{d}{k}} \,,
$$
where the constant $c$ depends on $(L, d, k)$. If $k > d/2$ then the function class is Donsker (i.e. it satisfies \eqref{eq:exp_growth} with $0 < \alpha < 2$), otherwise it is non-Donsker. By Proposition \ref{prop:rc_cover}, the Rademacher complexity of $\cF$ can be bounded as:
$$
\cR_n(\cF) \le
\begin{cases}   
Cn^{-\frac{1}{2}}\,, & \text{ if } \alpha > d/2, \\
Cn^{-\frac{\alpha}{d}}\,, & \text{ if } \alpha < d/2. 
\end{cases}
$$
Therefore, for the one-step procedure, Theorem \ref{thm:one_step} yields the following bound: 
$$
\bbE[\hat f_\pre (X)\mid \cD] - \bbE[f_0(X)] \lesssim
\begin{cases}
    \sqrt{\frac{t}{n}} \,, & \text{ if } \alpha > d/2, \\
    \frac{1}{n^{\alpha/d}} + \sqrt{\frac{t}{n}} \,, & \text{ if } \alpha < d/2 \,.
\end{cases}
$$
To ensure the coverage guarantee, we need to find the complexity of $\cG$ as defined in \eqref{eq:G_def}. It is immediate that $\cG \subseteq C_F^{\alpha}([0, 1]^d) \oplus C_{G}^{\beta}([0, 1]^{d+1})$ for some constant $G$. Therefore, the covering number of $\cG$ is bounded by the product of the covering numbers of $C_F^{\alpha}([0, 1]^d)$ and $C_G^{\beta}([0, 1]^{d+1})$ and consequently $\cG$ is a Donsker class if $0 < \gamma := \max\{d/\alpha ,(d+1)/\beta\}\leq 2$ and non-Donsker otherwise. This observation, along with Proposition \ref{prop:rc_cover} and Theorem \ref{thm:one_step} yields with probability at least $1 - e^{-t}$:   
\begin{align}
 \bbP\left((Y - \hat m(X))^2 \le \hat f_\pre(X) + \delta \,\big|\,\cD\right) & \geq 1-  \frac{C}{\delta\sqrt{n}}-\sqrt{\frac{2t}{n}} \hspace{0.5in} \text{if }0 < \gamma < 2 \,,\\
 \bbP\left((Y - \hat m(X))^2 \le \hat f_\pre(X) + \delta \,\big|\,\cD\right)  & \geq 1-  \frac{C}{\delta n^{\frac{1}{\gamma}}}-\sqrt{\frac{2t}{n}} \hspace{0.5in} \text{if }\gamma > 2 \,.
\end{align}

\subsection{Proof of Corollary \ref{cor:DNN}}
Define $\trunc_{B_\cF}(\cF_\DNN)$ is the truncated version of the functions of $\cF_\DNN$ at level $B_{\cF}$. Here we assume fully connected neural networks unless mentioned otherwise. A bound on the VC dimension of $\cF_{\mathrm{DNN}}$ can be found in \cite[Lemma]{bartlett2019nearly}, i.e. 
$$
\VC (\cF_{\mathrm{DNN}}) \le C WL\log{W} \,,
$$
where $W$ is the total number of parameters. Furthermore, as truncation does not change the VC dimension, the same bound also holds for $\trunc_{B_\cF}(\cF_\DNN)$. For a fully connected neural network with width $N$ and depth $L$, the total number of parameters is $N(d+1) + (L-2)(N^2+N) + 1$ (including bias) which is of order $N^2L$ for fixed $d$. Therefore, the VC dimension of $\cF_{\mathrm{DNN}}$ is of the order $N^2L^2\log{NL}$. Furthermore, we may also truncate $f$ at level $B_\cF$ to refrain it from being unbounded. Hence, from Theorem \ref{thm:two_step} and using the fact the that Rademacher complexity of a VC class is bounded by $\sqrt{V/n}$, we obtain:
\begin{equation}
\label{eq:bound_width_dnn}
\bbE[\hat f_\pre(X) \mid \cD] - \bbE[f_0(X)] \leq \Delta(\trunc_{B_\cF}(\cF_\DNN)) + C \frac{NL\sqrt{\log{NL}}}{\sqrt{n}} + C\|\hat m-m_0\|_{\infty}+4B_\cF\sqrt{\frac{t}{2n}}
\end{equation}
Now focus on $\Delta(\trunc_{B_\cF}(\cF_\DNN))$, we have: 
\begin{align*}
    \Delta(\trunc_{B_\cF}(\cF_\DNN)) & = \min_{\substack{f \in \trunc_{B_\cF}(\cF_\DNN) \\ f \ge f_0 \text{ a.s.}}} \bbE[f(X)] - \bbE[f_0(X)] \\
    & \le \min_{\substack{f \in \trunc_{B_\cF}(\cF_\DNN) \\ f \ge f_0}} \|f - f_0\|_\infty \\
    & \le \min_{\substack{f \in \cF_\DNN \\ f \ge f_0}} \|f - f_0\|_\infty \\
    & \le 2\min_{f \in \cF_\DNN} \|f - f_0\|_\infty \,.
\end{align*}
To understand how the last inequality follows, define $f^* = \argmin_{f \in \cF_\DNN}\|f - f_0\|_\infty$ and define $\delta = \| f^* - f_0\|_\infty$. Even if it may be possible $f^* < f_0$, we have $f^* + \delta \ge f_0$ almost surely. As $f^\star + \delta \in \cF_\DNN$ (adding a constant does not change the architecture), we have: 
$$
\min_{\substack{f \in \cF_\DNN \\ f \ge f_0}} \|f - f_0\|_\infty \le \|f^* + \delta - f_0\|_\infty \le 2 \delta=2 \min_{f \in \cF_\DNN}\|f - f_0\|_\infty\,.
$$
Now, if $f_0$ is in Hölder/Sobolev class of functions, i.e. $\alpha$ times differentiable with bounded derivatives, then \cite[Theorem 1]{lu2021deep} obtains the following approximation error: 
$$
\min_{f \in \cF_\DNN}\|f - f_0\|_\infty \le C_1 \left(\frac{NL}{\log{N}\log{L}}\right)^{-\frac{2\alpha}{d}}
$$
Therefore, from \eqref{eq:bound_width_dnn} we have: 
$$
\bbE[\hat f_\pre (X)\mid \cD] - \bbE[f_0(X)] \leq 2C_1 \left(\frac{NL}{\log{N}\log{L}}\right)^{-\frac{2\alpha}{d}} + C \frac{NL\sqrt{\log{NL}}}{\sqrt{n}}+ C\|\hat m-m_0\|_{\infty}+ 4B_\cF\sqrt{\frac{t}{2n}}.
$$
Since the VC dimension of $\cF_{\mathrm{DNN}}$ is of the order $N^2L^2\log{NL}$, by Theorem \ref{thm:two_step}, we have
\begin{align*}
 \bbP\left((Y - \hat m(X))^2 \le \hat f_\pre(X)\,\big|\,\cD\right)   \geq 1-c\sqrt{\frac{N^2L^2\log{NL}}{n}}-c\sqrt{\frac{t}{n}}.
\end{align*}

\subsection{Proof of Corollary \ref{cor:rkhs}}
To establish Corollary \ref{cor:rkhs} we need to find $\cR_n(\cF)$ and $\cR_n(\cG)$ in light of Theorem \ref{thm:one_step}. Recall that here $\Phi(\cdot)$ is the feature map, which maps $x \mapsto K(x, \cdot)$. As we have assumed $\sup_x K(x,x) \le b$, we have for any $h$ in the RKHS generated by $K$: 
\begin{equation}
\label{eq:bound_h}
|h(x)| = |\langle h, K(x, \cdot)\rangle_\cH| \le \|h\|_H\|K(x, \cdot)\|_\cH \le \sqrt{b}\|h\|_{\cH} \,. 
\end{equation}
In other words, $\|h\|_\infty \le \sqrt{b}\|h\|_\cH$. Furthermore, for any $f \in \cF$: 
\begin{align}
    \label{eq:bound_h_f}
    |f(x)| = \left|\left\langle \Phi(x), \cA (\Phi)(x)\right \rangle_\cH\right| & = \left|\left\langle \cA, \Phi(x) \otimes \Phi(x)^\top\right\rangle\right| \notag\\
    & \le \|\cA\|_\star \|\Phi(x) \otimes \Phi(x)^\top\|_{\mathrm{op}} \notag\\
    & \le r\sup_{h \in \cH: \|h\|_\cH \le 1} \left|\left\langle h, \Phi(x) \otimes \Phi(x)^\top(h) \right \rangle_\cH\right| \notag\\
    & \le r\sup_{h \in \cH: \|h\|_\cH \le 1} |h^2(x)| \le rb \,.
\end{align}
The first inequality follows from duality of operator and nuclear norm, second inequality follows from the definition of operator norm and the last inequality is from \eqref{eq:bound_h}. Therefore, $\|f\|_\infty \le rb$ for any $f \in \cF$. We now use these bounds to control the Rademacher complexity: 
\begin{align*}
    \cR_n(\cF) & = \bbE\left[\sup_{\|\cA\|_\star \le r}\left|\frac1n \sum_i \eps_i \left\langle \Phi(X_i), \cA (\Phi)(X_i)\right \rangle\right|\right] \\
    & \le r \bbE\left[\left\|\frac1n \sum_i \eps_i \Phi(X_i) \otimes \Phi(X_i)^\top \right\|_{\mathrm{op}}\right] \\
    & = r \bbE\left[\sup_{\|h \|_\cH \le 1} \left|\frac1n \sum_i \eps_i \langle \Phi(X_i), h \rangle^2 \right|\right] \\
    & = r \bbE\left[\sup_{\|h \|_\cH \le 1} \left|\frac1n \sum_i \eps_i h^2(X_i) \right|\right] 
\end{align*}
Next we apply Leduox-Talagrand contraction inequality (see Theorem 7 of \cite{duchi2009probability}) as $\|h\|_\infty \le \sqrt{b}$ by \eqref{eq:bound_h} and $\phi(x) = x^2$ is Lipschitz function on $[-\sqrt{b}, \sqrt{b}]$ with Lipschitz co-efficient $2\sqrt{b}$. Therefore we have: 
\begin{align*}
    \cR_n(\cF) & \le 4r\sqrt{b} \ \bbE\left[\sup_{\|h \|_\cH \le 1} \left|\frac1n \sum_i \eps_i h(X_i) \right|\right] \\
    & = 4r\sqrt{b} \ \bbE\left[\sup_{\|h \|_\cH \le 1} \left|\langle h, \frac1n \sum_i \eps_i K(X_i, \cdot) \rangle \right|\right] \\
    & \le 4r\sqrt{b} \ \bbE\left[\left\| \frac1n \sum_i \eps_i K(X_i, \cdot)\right\|_\cH\right] \\
    & \le 4r\sqrt{b} \ \sqrt{\bbE\left[\left\| \frac1n \sum_i \eps_i K(X_i, \cdot)\right\|^2_\cH\right]} \\
    & = 4r\sqrt{b} \ \sqrt{\bbE\left[\frac{1}{n^2}\sum_{i, j} \eps_i \eps_j K(X_i, X_j)\right]} \\
    & = \frac{4r\sqrt{b}}{n} \ \sqrt{\bbE\left[\tr{K}\right]} = \frac{4r\sqrt{b}}{\sqrt{n}}\sqrt{\bbE\left[K(X, X)\right]} \,. 
\end{align*}
Combining the above inequality with Theorem \ref{thm:one_step}, we get the bound of the obtained bandwidth. The coverage guarantee directly follows from the observation that $\cR_n(\cG)\lesssim\cR_n(\cM)+\cR_n(\cF)$.

\end{document}